\documentclass[12pt,report]{amsart}
\usepackage{graphicx}
\usepackage[all,cmtip]{xy}
\usepackage[margin=1in]{geometry}
\usepackage{amssymb}
\usepackage{amsthm}
\usepackage{hyperref}
\usepackage{multicol}
\usepackage{latexsym}
\usepackage{amsfonts}
\usepackage{amscd}
\usepackage[all]{xy}
\usepackage{slashed}
\usepackage{todonotes}
\let\oldtocsection=\tocsection

\let\oldtocsubsection=\tocsubsection

\let\oldtocsubsubsection=\tocsubsubsection

\renewcommand{\tocsection}[2]{\hspace{0em}\oldtocsection{#1}{#2}}
\renewcommand{\tocsubsection}[2]{\hspace{1em}\oldtocsubsection{#1}{#2}}
\renewcommand{\tocsubsubsection}[2]{\hspace{2em}\oldtocsubsubsection{#1}{#2}}

\newtheorem{proposition}{Proposition}
\newtheorem{theorem}[proposition]{Theorem}

\newtheorem{lemma}[proposition]{Lemma}

\newtheorem{define}[proposition]{Definition}
\newtheorem{rem}[proposition]{Remark}

\newcommand{\Z}{\mathbb{Z}}
\newcommand{\R}{\mathbb{R}}
\newcommand{\C}{\mathbb{C}}

\newcommand{\Q}{\mathbb{Q}}
\newcommand{\U}{\mathbb{U}}

\makeatletter
\newcommand{\xyC}[1]{%
\makeatletter
\xydef@\xymatrixcolsep@{#1}
\makeatother
} 

\makeatletter
\newcommand{\xyR}[1]{%
\makeatletter
\xydef@\xymatrixrowsep@{#1}
\makeatother
} 

\def\Xint#1{\mathchoice
{\XXint\displaystyle\textstyle{#1}}%
{\XXint\textstyle\scriptstyle{#1}}%
{\XXint\scriptstyle\scriptscriptstyle{#1}}%
{\XXint\scriptscriptstyle\scriptscriptstyle{#1}}%
\!\int}
\def\XXint#1#2#3{{\setbox0=\hbox{$#1{#2#3}{\int}$}
\vcenter{\hbox{$#2#3$}}\kern-.5\wd0}}
\def\ddashint{\Xint=}

\begin{document}
\title{Localization In Abelian Chern-Simons Theory}

\author{B.D.K. McLellan}

\begin{abstract}
Chern-Simons theory on a closed contact three-manifold is studied when the Lie group for gauge transformations is compact, connected and abelian.  The abelian Chern-Simons partition function is derived using the Faddeev-Popov gauge fixing method.  The partition function is then formally computed using the technique of non-abelian localization.  This study leads to a natural identification of the abelian Reidemeister-Ray-Singer torsion as a specific multiple of the natural unit symplectic volume form on the moduli space of flat abelian connections for the class of Sasakian three-manifolds.
The torsion part of the abelian Chern-Simons partition function is computed explicitly in terms of Seifert data for a given Sasakian three-manifold.
\end{abstract}




\maketitle
\tableofcontents

\section{Introduction}\label{nonabelchap}
The goal of this article is to study the Chern-Simons partition function as a rigorous topological three manifold invariant using heuristic techniques that arise in the physics literature.  Our starting point is a heuristically defined partition function, defined as a Feynman path integral, that physically describes the quantum amplitude of a given three manifold $X$ with respect to the Chern-Simons action.  From a mathematical perspective the partition function is somewhat mysterious and a rigorous, systematic method for its study, in the sense of \emph{constructive quantum field theory}, is currently lacking.  It is remarkable, however, that one can make rigorous mathematical predictions using the partition function.\\
\\
Our motivation for this study is contained primarily in the work of C. Beasley and E. Witten, \cite{bw}, where the Chern-Simons partition function is studied with respect to a contact structure $H\subset TX$ on a three manifold $X$.  The study of Chern-Simons theory with respect to a contact structure is a novel idea that Beasley and Witten introduce in order to explain some empirical observations of L. Rozansky, \cite{roz}.  Rozansky studied Chern-Simons theory on Seifert manifolds and observed that the contributions from irreducible flat connections were finite loop exact.  He then further observed that this was similar to the behaviour of two dimensional Yang-Mills theory, where similar phenomenon are explained by non-abelian localization \cite{w2}.\\
\\
Our main idea in this article is to study the \emph{abelian} partition function with respect to a contact structure following the ideas of \cite{bw}.  We note that the first systematic work on Chern-Simons theory in the physics literature, including in particular the necessary
quantization of its coefficient, may be found in the work of S. Deser, R. Jackiw, G. 't Hooft and S. Templeton \cite{djt}.  The abelian theory is well known and has been studied from several different perspectives.  Recall, $G=\operatorname{U}(1)$ Chern-Simons theory is physically interesting and can be realized as a fundamental building block for a theory of the fractional quantum Hall effect, \cite{wilcz}, \cite{zeefro}.  Schwarz \cite{s1}, \cite{s2} has also shown that the $\operatorname{U}(1)$ theory is related to the Reidemeister-Ray-Singer torsion \cite{reid2}, \cite{rsi}, a classical topological invariant of three manifolds.  Abelian Chern-Simons theory is also closely related to the one loop contribution of Beasley and Witten's work, which is also a motivation for this study.\\
\\
Recall that Witten has shown \cite{w3} that the one loop contribution to the partition function requires some of the work of Atiyah, Patodi and Singer, \cite{aps1}, \cite{aps2}, \cite{aps3}, in order to extract a topological invariant of a three manifold.  A generalization of our study will naturally involve making sense of an analogue of some of the work of Atiyah, Patodi, and Singer for three manifolds relative to an arbitrary contact structure, going beyond the Seifert case.\\
\\
Recall, \cite{bw} studies the Chern-Simons partition function, \cite[Eq. 3.1]{bw}, which is heuristically defined as follows,
\begin{equation}\label{orgchern}
Z(k):=\frac{1}{\operatorname{Vol}(\mathcal{G})}\left(\frac{k}{4\pi^2}\right)^{\Delta{\mathcal{G}}} \ddashint \mathcal{D}A\,\,e^{ik\operatorname{CS}(A)}.
\end{equation}
\begin{rem}
The notation, $\ddashint$, is introduced to explicitly distinguish a ``path integral'' from ordinary integration.  The notation, $\ddashint$, describes an ``integral'' over the the ``space of connections,'' and serves as a heuristic device that is generally non-rigorous.  Quantities involving integrals over the moduli space of flat connections $\mathcal{M}_{P}$ are rigorously defined and derived from the former explicitly in this article.
\end{rem}
In general, the partition function of Eq. \eqref{orgchern} is not known to admit a general mathematical interpretation in terms of the cohomology of some classical moduli space of connections, in contrast to Yang-Mills theory for example, \cite{w2}.  The main result of \cite{bw}, however, is that if $X$ is assumed to carry the additional geometric structure of a Seifert manifold, then the partition function of Eq. \eqref{orgchern} \emph{does} admit a more conventional interpretation in terms of the cohomology of some classical moduli space of connections.
Using the additional Seifert structure on $X$, \cite{bw} decouple one of the components of a gauge field $A$, and introduce a ``new'' partition function denoted by $\bar{Z}(k)$ and given as \cite[Eq. 3.7]{bw},
\begin{equation}\label{newchern}
K\cdot\ddashint \mathcal{D}A\mathcal{D}\Phi\,\,\operatorname{exp}\left[ik\left(\operatorname{CS}(A)-\frac{1}{4\pi}\int_{X}2\kappa\wedge\operatorname{Tr}(\Phi F_{A})+\frac{1}{4\pi}\int_{X}\kappa\wedge d\kappa\,\,\operatorname{Tr}(\Phi^{2})\right)\right],
\end{equation}
where the basic ingredients in this expression are given in \S\ref{shiftsec}.  \cite{bw} give a heuristic argument showing that the partition function computed using the alternative description of Eq. \eqref{newchern} should be the same as the Chern-Simons partition function of Eq. \eqref{orgchern}.  In essence, they show, \cite[pg.13]{bw},
\begin{equation}\label{maint}
Z(k)=\bar{Z}(k),
\end{equation}
by gauge fixing $\Phi=0$ using the shift symmetry.  \cite{bw} then observe that the $\Phi$ dependence in the integral can be eliminated by simply performing the Gaussian integral over $\Phi$ in Eq. \eqref{newchern} directly.  They obtain the alternative formulation,
\begin{equation}\label{newchern2}
Z(k)=\bar{Z}(k)=K'\cdot\ddashint \mathcal{D}A\,\,\operatorname{exp}\left[ik\left(\operatorname{CS}(A)-\frac{1}{4\pi}\int_{X}\frac{1}{\kappa\wedge d\kappa}\,\,\operatorname{Tr}\left[(\kappa\wedge F_{A})^{2}\right]\right)\right],
\end{equation}
where $K':=\frac{1}{\operatorname{Vol}(\mathcal{G})}\frac{1}{\operatorname{Vol}(\mathcal{S})}\left(\frac{-ik}{4\pi^2}\right)^{\Delta{\mathcal{G}}/2}$.  Note that we follow \cite[Eq. 3.9]{bw} here and abuse notation slightly by writing $\frac{1}{\kappa\wedge d\kappa}$.  We have done this with the understanding that since $\kappa\wedge d\kappa$ is non-vanishing (since $\kappa$ is a contact form), then $\kappa\wedge F_{A}=\phi\,\kappa\wedge d\kappa$ for some function $\phi\in \Omega^{0}(X,\frak{g})$, and we identify $\frac{\kappa\wedge F_{A}}{\kappa\wedge d\kappa}:=\phi$.\\
\\
The original argument of \cite{bw} was to decouple one of the components of the gauge field $A\in \mathcal{A}_{P}$ by introducing a \emph{local shift symmetry} (see \cite[\S3.1]{bw}) and then to translate the Chern-Simons partition function into a ``moment map squared'' form using this symmetry.  The general ``moment map squared'' form for the partition function is a symplectic integral of the canonical form,
\begin{equation}\label{canonform}
\bar{Z}(\epsilon)=\frac{1}{\operatorname{Vol}(H)}\left(\frac{1}{2\pi\epsilon}\right)^{\Delta_{H}/2}\ddashint_{Y}\exp\left[\Omega-\frac{1}{2\epsilon}(\mu,\mu)\right],
\end{equation}
where $Y$ is a symplectic manifold with symplectic form $\Omega$, and $H$ is a Lie group that acts on $Y$ in a Hamiltonian fashion with moment map $\mu$.  $\Delta_{H}=\operatorname{dim}(H)$ and $\epsilon=\frac{2\pi}{k}$.  The technique of non-abelian localization \cite{w2} can then be applied to study such integrals.  This article studies the analogous theory in the case of a compact, connected abelian Lie group $G$.\\
\\
The main goal of this article is to study an analogue of the partition function in \eqref{newchern} for the case of a compact, connected and abelian structure group, and to use the equivalence \eqref{maint} to compute the abelian partition function.
\\
\\
In \S \ref{u1partsec} we derive a definition of the abelian partition function, $Z_{\mathbb{T}}(X,P,k)$, which is the abelian analogue of \eqref{orgchern} and is the main topological invariant studied in this article.  Our derivation starts with a heuristically defined partition function, as introduced in remark \ref{heurdef},
$$Z_{\mathbb{T}}(X,P,k)=\frac{1}{\operatorname{Vol}(\mathcal{G})} \ddashint_{\mathcal{A}_{P}} \mathcal{D}A\, e^{i k \operatorname{CS}_{X,P}(A_{P})},$$
where we abuse notation and write $Z_{\mathbb{T}}(X,P,k)$ for both the heuristic and rigorous versions of the partition function.  Our method uses the ``Faddeev-Popov gauge fixing method,'' as introduced in \cite{fadpop}, to extract a rigorously defined topological invariant.  We note that our method differs from previous derivations \cite{s2}, \cite{m} of a rigorous candidate for an abelian Chern-Simons partition function.  Our approach generalizes more readily to non-abelian gauge groups and also leads to a different and more natural $k$-dependence for the rigorous quantity we obtain.  We also note that our approach differs from that of \cite{m} in that we take into account a dependence of the partition function on a choice of two-framing on $X$.  Our final rigorous definition of the abelian Chern-Simons partition function is given in definition \ref{rigpartdef}.\\
\\
In \S\ref{shiftsec} we study the ``shift symmetry'' construction introduced in \cite{bw}, and apply this to the heuristic abelian partition function.  The main result of this section is a ``new'' heuristic definition of what we call the \emph{shift reduced abelian Chern-Simons partition function} given in equation \eqref{Anom12}.  In \S\ref{momsquared} our main objective is to present the shift reduced abelian Chern-Simons partition function in the canonical moment map squared form as in \eqref{canonform}.  Following the basic argument of \cite{bw}, we are able to obtain this result in equation \eqref{Anom1234}.  This allows us to formally apply the method of non-abelian localization to the heuristic path integral in \S\ref{locabelchern} to finally obtain a ``new'' rigorous definition of the partition function as in definition \ref{symrigdef}.\\
\\
In summary, we make two rigorous definitions in the article, the \emph{abelian Chern-Simons partition function} $Z_{\mathbb{T}}(X,k)$ in definition \ref{rigpartdef} and the \emph{symplectic abelian Chern-Simons partition function} $\bar{Z}_{\mathbb{T}}(X,k)$ in definition \ref{symrigdef}, using heuristic techniques.  We find that both $Z_{\mathbb{T}}(X,k)$ and $\bar{Z}_{\mathbb{T}}(X,k)$ have identical $k$-dependent terms $k^{m_{X}}$.  We note that our $k$-dependence differs from that in \cite{m} since we take into account the $k$-dependence due to the isotropy group $I$.  Our computation physically identifies the volume form dependent parts of $Z_{\mathbb{T}}(X,k)$ and $\bar{Z}_{\mathbb{T}}(X,k)$.  Futhermore, one can prove the following,
\begin{theorem}\label{moduliprop2}
Given a closed, oriented Seifert three manifold $X$ such that $c_{1}(X)\neq 0$ (where $c_{1}(X)$ denotes the first orbifold Chern number of $X$) then,
$$\mathcal{M}_{X}\simeq \mathbb{T}^{2g}\times \operatorname{Tors}(H^{2}(X,\Lambda))\simeq \operatorname{Hom}(\pi_{1}(X),\mathbb{T}),$$
where, $|\operatorname{Tors}H^{2}(X,\Lambda)|=|c_{1}(X)\cdot\prod_{j=1}^{M}\alpha_{j}|^{N}$, and $\Lambda$ denotes the integral lattice of the Lie group $\mathbb{T}$.
\end{theorem}
\noindent
Using theorem \ref{moduliprop2}, one may compute,
\begin{equation}
\sqrt{T_{X}}=\frac{\omega_{P}}{|c_{1}(X)\cdot \prod_{i}\alpha_{i}|^{N/2}}=\frac{\omega_{P}}{\sqrt{|\operatorname{Tors}H^{2}(X,\Lambda)|}}.
\end{equation}
We also note that previous work also identifies the eta-invariant dependent parts of $Z_{\mathbb{T}}(X,k)$ and $\bar{Z}_{\mathbb{T}}(X,k)$.  This argument is summarized in \cite{jm2} and uses the main result of \cite{mcl2}.  We note that the main result of \cite{mcl2} assumes a natural choice of \emph{Seifert two-framing} on $X$ and does not study how the partition function changes under a change in two-framing associated to a change in the underlying contact structure.  We leave this to future work.\\
\\
Overall, this article mathematically defines the quantities $Z_{\mathbb{T}}(X,k)$, $\bar{Z}_{\mathbb{T}}(X,k)$, and \emph{physically} computes the magnitude,
\begin{equation}
\left|Z_{\mathbb{T}}(X,k)\right|=k^{m_X}\cdot \frac{\left|\sum_{[P]\in\operatorname{Tors}H^{2}(X,\Lambda)}e^{i k \operatorname{CS}_{X,P}(A_{P})}\right|}{\sqrt{|\operatorname{Tors}H^{2}(X,\Lambda)|}}.
\end{equation}
Lastly, we note that it would be of interest to compute the quantity,
\begin{equation}
\left|\sum_{[P]\in\operatorname{Tors}H^{2}(X,\Lambda)}e^{i k \operatorname{CS}_{X,P}(A_{P})}\right|,
\end{equation}
explicitly in terms of Seifert data on $X$, and indeed it would be interesting to make an explicit computation of this quantity on a general closed three-manifold.

\section{The Abelian Partition Function}\label{u1partsec}
In this section we define a partition function, $Z_{\mathbb{T}}(X,k)$, for abelian Chern-Simons theory.  A closely related partition function is studied in \cite{m}, where it shown that it defines a unitary topological quantum field theory as defined by Atiyah in \cite{attqft}.  Our definition of the abelian Chern-Simons partition function differs from \cite[Eq. 7.28]{m} in that we take into account a dependence of the partition function on a choice of two-framing on $X$.  We follow \cite{w3} and revise the definition of \cite{m} by adding a ``counterterm,'' the gravitational Chern-Simons term, to the eta-invariant that shows up in our considerations.  By an Atiyah-Patodi-Singer theorem \cite[Prop. 4.19]{aps2}, this counterterm effectively restores topological invariance for the partition function.  We also choose a different $k$-dependence for the partition function than \cite{m} in order to reflect a dependence of the isotropy group on $k$.\\
\\
Before we define the partition function, we establish some notation and terminology.  Let $\mathbb{T}$ denote a compact, connected abelian Lie group of dimension $N$, $\frak{t}$ denote its Lie algebra and $\Lambda\subset \frak{t}$ the integral lattice.  Let $\operatorname{Tors}H^{2}(X,\Lambda)$ denote the torsion subgroup of $H^{2}(X,\Lambda)$.  $\mathcal{A}_{P}$ is the affine space of connections on $P$ modeled on the vector space $\Omega^{1}(X,\frak{t})$.  $\mathcal{G}:=\operatorname{Map}(X,\mathbb{T})$ is the group of gauge transformations and acts on $\mathcal{A}_{P}$ in the standard way.  That is, for $g\in\operatorname{Map}(X,\mathbb{T})$, and $A_{P}\in \mathcal{A}_{P}$, $A_{P}\cdot g:=A_{P}+g^{*}\vartheta$, where $\vartheta\in \Omega^{1}(\mathbb{T},\frak{t})$ denotes the Maurer-Cartan form on $\mathbb{T}$.  $\operatorname{CS}_{X,P}(A_{P})$ is the Chern-Simons functional of a $\mathbb{T}$-connection $A_{P}$ on $P\rightarrow X$ and we describe this presently.  For any $\mathbb{T}$-connection $A_{P}\in\mathcal{A}_{P}$, we define an $\operatorname{SU}(N+1)$-connection $\hat{A}_{P}$ on an associated principal $\operatorname{SU}(N+1)$-bundle,
\begin{equation}\label{hatbundle}
\hat{P}=P\times_{\mathbb{T}} \operatorname{SU}(N+1),
\end{equation}
via, $$\hat{A}_{P}|_{[p,h]}=\operatorname{Ad}_{h^{-1}}(\iota_{*}\operatorname{pr}_{1}^{*}A_{P}|_{p})+\operatorname{pr}_{2}^{*}\vartheta_{h},$$ where $\iota:\mathbb{T}\rightarrow \operatorname{SU}(N+1)$ is inclusion as a maximal torus, $\operatorname{pr}_{1}:P\times \operatorname{SU}(N+1)\rightarrow P$ and $\operatorname{pr}_{2}:P\times \operatorname{SU}(N+1)\rightarrow \operatorname{SU}(N+1)$ are the standard projections.
Since for any three manifold $X$, $\hat{P}$ is trivializable, let $\hat{s}:X\rightarrow \hat{P}$ be a global section.  The definition we use for the Chern-Simons action, $\operatorname{CS}_{X,P}(A_{P})$, is as follows,
\begin{define}\label{csactiondef}
The Chern-Simons action functional of a $\mathbb{T}$-connection $A_{P}\in\mathcal{A}_{P}$ is defined by,
\begin{equation}\label{act}
\operatorname{CS}_{X,P}(A_{P}):=\frac{1}{4\pi}\int_{X}\hat{s}^{*}\alpha(\hat{A}_{P})\,\,\,\text{mod}\,\, (2\pi\Z),
\end{equation}
where $\alpha(\hat{A}_{P})\in\Omega^{3}(\hat{P},\R)$ is the Chern-Simons form of the induced $\operatorname{SU}(N+1)$-connection $\hat{A}_{P}\in\mathcal{A}_{\hat{P}}$,
\begin{equation}
\alpha(\hat{A}_{P}):=\operatorname{Tr}(\hat{A}_{P}\wedge F_{\hat{A}_{P}})-\frac{1}{6}\operatorname{Tr}(\hat{A}_{P}\wedge [\hat{A}_{P},\hat{A}_{P}]),
\end{equation}
where $\operatorname{Tr}:\frak{su}(N+1)\otimes\frak{su}(N+1)\rightarrow \R$ denotes the standard $\operatorname{Ad}$-invariant bilinear form in the $(N+1)$-dimensional representation.
\end{define}
\begin{rem}\label{heurdef}
One may then \emph{heuristically} define a ``partition function'' as follows.
Let $k\in\Z$ and $X$ a closed, oriented three-manifold.  The \emph{abelian Chern-Simons partition function}, $Z_{\mathbb{T}}(X,k)$, is the heuristic quantity,
\begin{equation}\label{hpart1}
Z_{\mathbb{T}}(X,k)=\sum_{P\in\operatorname{Tors}H^{2}(X,\Lambda)}Z_{\mathbb{T}}(X,P,k),
\end{equation}
and,
\begin{equation}\label{hpart2}
Z_{\mathbb{T}}(X,P,k)=\frac{1}{\operatorname{Vol}(\mathcal{G})} \ddashint_{\mathcal{A}_{P}} \mathcal{D}A\, e^{i k \operatorname{CS}_{X,P}(A_{P})}.
\end{equation}
Note that $\operatorname{Vol}(\mathcal{G})$ formally denotes the volume of the gauge group.
\end{rem}
In our heuristic definition of the abelian partition in the above remark \eqref{heurdef} we sum over flat bundle classes corresponding to elements of $\operatorname{Tors}H^{2}(X,\Lambda)$ because these stationary points of the Chern-Simons action are the flat connections and these are precisely the bundles that admit flat connections.  Note that Eq. \eqref{hpart2} is a formal expression, where we heuristically assume the existence of the measure $\mathcal{D}A$.  It is precisely the quantity $\frac{\mathcal{D}A}{\operatorname{Vol}(\mathcal{G})}$ in Eq. \eqref{hpart2} that is \emph{not} well defined.  Our goal in this section is to make definition \eqref{heurdef} rigorous using the Faddeev-Popov method \cite{fadpop}.
We recall the main ingredients that go into the heuristic definition of the partition function in \eqref{heurdef} above.  First, the measure $\mathcal{D}A$ is \emph{formally} induced by a choice of metric $\operatorname{g}$ on $X$.
Let $\langle\cdot,\cdot\rangle:\frak{t}\otimes\frak{t}\rightarrow\R$ be the bilinear form on $\mathbb{T}$ induced by $\operatorname{Tr}$ as in definition \ref{csactiondef}.  Then $\operatorname{g}$ defines the Hodge star operator, $\star$, on the tangent space $T_{A_{P}}\mathcal{A}_{P}\simeq \Omega^{1}(X,\frak{t})$, which in turn induces the $\mathcal{G}$-invariant Riemannian metric,
\begin{equation}\label{inneq1}
\langle A, B\rangle_{L^{2}}:=\int_{X}\langle A\wedge\star B\rangle,
\end{equation}
on $\mathcal{A}_{P}$, for $A,B\in T_{A_{P}}\mathcal{A}_{P}\simeq \Omega^{1}(X,\frak{t})$.
Observe that,
\begin{eqnarray}\label{quadfun}
Z_{\mathbb{T}}(X,P,k)&=&\frac{e^{i k \operatorname{CS}_{X,P}(A_{P})}}{\operatorname{Vol}(\mathcal{G})} \ddashint_{\mathcal{A}_{P}}\mathcal{D}A\,\, \exp\,\left[\frac{i k}{4\pi}\left(\int_{X} \langle A\wedge dA\rangle\right)\right],
\end{eqnarray}
where we rewrite the partition function after identifying $\mathcal{A}_{P}=A_{P}+\Omega^{1}(X,\frak{t})$ for a flat base point $A_{P}$ in $\mathcal{A}_{P}$.
We then use the Faddeev-Popov method \cite{fadpop} to obtain an exact result.  We will gauge fix by choosing a metric $\operatorname{g}$ on $X$ and fix the \emph{Lorenz gauge condition},
\begin{equation}\label{gaugecond1}
\operatorname{C}(A)=d^{\dagger}A=0,
\end{equation}
where on $q$-forms $d^{\dagger}=(-1)^{q+1}\star d \star$ is the adjoint of $d$ defined with respect to the Hodge star $\star$ for the metric $\operatorname{g}$.  We introduce the gauge fixing action,
\begin{equation}
\operatorname{S}_{\operatorname{gauge}}(A,\Psi,\bold{c},\bar{\bold{c}})=\frac{1}{2\pi}\int_{X}\left(\langle d^{\dagger}A\wedge \Psi\rangle+\langle\bar{\bold{c}}d^{\dagger}d\bold{c}\rangle \right),
\end{equation}
where $\Psi\in\Omega^{3}(X,\frak{t})$ is a Lagrange multiplier term that enforces the gauge condition \eqref{gaugecond1}, and $\bold{c}, \bar{\bold{c}}$ are formal anti-commuting Lie algebra valued \emph{ghost fields} that allow one to write the measure fixing determinant of $d^{\dagger}d$ in exponential form.
Let $I$ denote the isotropy subgroup of $\mathcal{G}$ at $A_{P}\in \mathcal{A}_{P}$.  This is the group of constant maps from $X$ to $\mathbb{T}$ since,
\begin{equation}
\theta\in\text{Lie}\,\mathcal{G}:A_{P}\mapsto A_{P}+d\theta,
\end{equation}
I.e. $d\theta=0\Rightarrow \theta=$ constant (we assume that $X$ is connected), and hence $I\simeq \mathbb{T}$.  Let $\operatorname{Vol}I$ be the volume of the isotropy subgroup with respect to the induced measure on $\mathcal{G}$,
\begin{equation}\label{isovol}
\operatorname{Vol}I=[\operatorname{Vol}X]^{N/2}=\left[\int_{X}\star 1\right]^{N/2}.
\end{equation}
Eq. \eqref{isovol} follows from the definition of the invariant metric on the group $\mathcal{G}$ that is induced by the inner product on $\text{Lie}\,\mathcal{G}\simeq \Omega^{0}(X,\frak{t})$ that comes from $\operatorname{g}$,
\begin{equation}
G_{\mathcal{G}}(\theta,\phi):=\int_{X}\langle\theta\wedge\star \phi\rangle,
\end{equation}
where $\theta,\phi\in \text{Lie}\,\mathcal{G}\simeq \Omega^{0}(X,\frak{t})$.
Observe that $G_{\mathcal{G}}$ restricted to the space of constant functions is simply a scalar multiple of $\langle\cdot,\cdot\rangle_{L^{2}}$ at each $\Psi\in\mathcal{G}\simeq \operatorname{Map}(X,\mathbb{T})$,
\begin{eqnarray*}
G_{\mathcal{G}}(\theta,\phi)|_{\Psi}&=&\int_{X}\langle\theta\wedge\star \phi\rangle,\\
                     &=&\left(\int_{X}\star 1\right)\cdot\langle\theta,\phi\rangle,
\end{eqnarray*}
since $\theta,\phi\in\frak{t}$ are constant.  We may therefore write $\sqrt{G_{\mathcal{G}}}=\left(\int_{X}\star 1\right)^{N/2}$.  If $\sqrt{G_{\mathcal{G}}}D\sigma$ denotes the measure on $I<\mathcal{G}$, then,
\begin{eqnarray}\nonumber
\operatorname{Vol}I&=&\int_{\mathbb{T}}\sqrt{G_{\mathcal{G}}}D\sigma,\\\nonumber
                &=&\sqrt{G_{\mathcal{G}}},\,\text{setting $\int_{I}D\sigma=1$,}\\\label{volres}
                &=&\left[\int_{X}\star 1\right]^{N/2}.
\end{eqnarray}
Observe that in \eqref{quadfun} we may integrate the gauge orbit out and write our integral over the quotient space $\mathcal{A}_{P}/\mathcal{G}$. Note that the metric on $\mathcal{A}_{P}$ formally descends to a metric on the quotient $\mathcal{A}_{P}/\mathcal{G}$, and thereby induces a quotient measure that we denote by $\widehat{\mathcal{D}A}$.  The integral over of the gauge orbit will contribute a factor of,
$$\frac{\operatorname{Vol}\mathcal{G}}{\operatorname{Vol}I},$$
due to the presence of the isotropy group.
We now \emph{define} $Z_{\mathbb{T}}(X,P,k)$ as,
\begin{equation}\label{1looppart}
K(A_{P},k)\cdot\ddashint\mathcal{D}A\mathcal{D}\Psi\mathcal{D}\mathbf{c}\mathcal{D}\bar{\mathbf{c}}\,\,\operatorname{exp}\left[\frac{ik}{4\pi}\int_{X}\langle A\wedge dA\rangle+ik\operatorname{S}_{\operatorname{gauge}}(A,\Psi,\bold{c},\bar{\bold{c}})\right],
\end{equation}
where,
\begin{equation}\label{extrafac}
K(A_{P},k):=\frac{e^{i k \operatorname{CS}_{X,P}(A_{P})}}{\operatorname{Vol}I}\cdot k^{-\frac{1}{2}\operatorname{dim}H^{0}(X,\frak{t})},
\end{equation}
and we have included the factor $k^{-\frac{1}{2}\operatorname{dim}H^{0}(X,\frak{t})}$ in \eqref{extrafac} to take into account the $k$-dependence that occurs in the volume of the isotropy group as in \eqref{volres}.
Let,
$$\operatorname{L}:\Omega^{\bullet}(X,\frak{t})\rightarrow \Omega^{\bullet}(X,\frak{t}),$$
denote the self-adjoint operator defined by,
$$\operatorname{L}:=\star d+d\star,$$
and let $\operatorname{L}^{\operatorname{o}}$ denote the operator $\operatorname{L}$ restricted to the odd forms, $\Omega^{1}(X,\frak{t})\oplus \Omega^{3}(X,\frak{t})$.  Observe that the $(A,\Psi)$ dependent part of the action in \eqref{1looppart} may be expressed as,
$$\int_{X}\langle A\wedge dA+2d^{\dagger}A\cdot \Psi\rangle=\langle (A,\Psi),\operatorname{L}^{\operatorname{o}}(A,\Psi)\rangle_{L^{2}}.$$
Overall, \eqref{1looppart} leads to the following expression for $Z_{\mathbb{T}}(X,P,k)$,
\begin{eqnarray*}\nonumber
&&K(A_{P},k)\cdot\ddashint\mathcal{D}A\mathcal{D}\Psi\mathcal{D}\mathbf{c}\mathcal{D}\bar{\mathbf{c}}\,\,\operatorname{exp}\left[\frac{ik}{4\pi}\langle (A,\Psi),\operatorname{L}^{\operatorname{o}}(A,\Psi)\rangle_{L^{2}}+\frac{ik}{2\pi}\int_{X}\langle\bar{\bold{c}}d^{\dagger}d\bold{c}\rangle\right],\\\nonumber
                     &=&K'(A_{P},k)\cdot\ddashint\mathcal{D}A\mathcal{D}\Psi\,\,\operatorname{exp}\left[\frac{ik}{4\pi}\langle (A,\Psi),\operatorname{L}^{\operatorname{o}}(A,\Psi)\rangle_{L^{2}}\right]\operatorname{det}'\left[d^{\dagger}d\right],
\end{eqnarray*}
where $\mathbf{c},\bar{\mathbf{c}}$ have been integrated out to obtain the last line, and $\operatorname{det}'$ denotes a regularized determinant to be defined later.  As shown in Lemma \ref{scallem} below, the determinant $\operatorname{det}'$ defined in \eqref{regdet1} satisfies the scaling,
\begin{equation}\label{scalingeq}
\operatorname{det}'\left[c\cdot d^{\dagger}d\right]=c^{-\operatorname{dim}H^{0}(X,\frak{t})}\cdot\operatorname{det}'\left[d^{\dagger}d\right],
\end{equation}
for $c\in\R_{+}$.
We have therefore multiplied $K(A_{P},k)$ by the factor,
\begin{equation}\label{newfactor1}
k^{-\operatorname{dim}H^{0}(X,\frak{t})},
\end{equation}
and we have,
\begin{equation}
K'(A_{P},k):=\frac{e^{i k \operatorname{CS}_{X,P}(A_{P})}}{\operatorname{Vol}I}k^{-\frac{3}{2}\operatorname{dim}H^{0}(X,\frak{t})}.
\end{equation}
We may define the path integral,
\begin{equation}\label{alfin}
Z_{\mathbb{T}}(X,P,k)=K'(A_{P},k)\cdot\ddashint\mathcal{D}A\mathcal{D}\Psi\,\,\operatorname{exp}\left[\frac{ik}{4\pi}\langle (A,\Psi),\operatorname{L}^{\operatorname{o}}(A,\Psi)\rangle_{L^{2}}\right]\operatorname{det}'\left[d^{\dagger}d\right],
\end{equation}
formally using stationary phase.  Let $\operatorname{spec}^{*}(\operatorname{L}^{\operatorname{o}})$ denote the non-zero part of the spectrum of $\operatorname{L}^{\operatorname{o}}$ and formally define the signature of $\operatorname{L}^{\operatorname{o}}$,
$$\operatorname{sgn}(\operatorname{L}^{\operatorname{o}}):=\sum_{\lambda\in \operatorname{spec}^{*}(\operatorname{L}^{\operatorname{o}})}\operatorname{sign}\lambda,$$
where $\operatorname{sign}\lambda=\pm 1$ denotes the sign of the real number $\lambda$.  Of course, this expression for the signature of $\operatorname{L}^{\operatorname{o}}$ is not generally well defined and we will regularize using an eta-invariant to obtain something sensible later. 
Thus, applying stationary phase, we obtain,
\begin{equation}\label{stationdef}
Z_{\mathbb{T}}(X,P,k)=K'(A_{P},k)\cdot\int_{\mathcal{M}_{P}}\,\,\frac{1}{\sqrt{|\operatorname{det}'k\operatorname{L}^{\operatorname{o}}|}}\operatorname{exp}\left[\frac{i\pi}{4}\operatorname{sgn}(\operatorname{L}^{\operatorname{o}})\right]\operatorname{det}'\left[d^{\dagger}d\right]\,\nu,
\end{equation}
where $\nu$ denotes the natural measure on the moduli space of flat abelian connections on $P$, $\mathcal{M}_{P}\simeq H^{1}(X,\frak{t})/H^{1}(X,\Lambda)$, and
$\operatorname{det}'\operatorname{L}^{\operatorname{o}}$ is formally the product of non-zero eigenvalues of $\operatorname{L}^{\operatorname{o}}$.  We note that the Lagrange multiplier integration variable $\Psi\in\Omega^{3}(X,\frak{t})$ in \eqref{alfin} is such that $d^{\dagger}\Psi\neq 0$, and accordingly the integral localizes on the critical point set given by $\mathcal{M}_{P}$.  
\begin{rem}
Although the derivation of \eqref{stationdef} is standard, our method is different than that used by \cite{m}, which in turn uses the method of Schwarz \cite{s2}.
\end{rem}
We will define $\operatorname{det}'\operatorname{L}^{\operatorname{o}}$ via regularization using a zeta function determinant.  First, we regularize the signature $\operatorname{sgn}(\operatorname{L}^{\operatorname{o}})$ via the eta-invariant and set $\operatorname{sgn}(\operatorname{L}^{\operatorname{o}})\rightsquigarrow \eta(\operatorname{L}^{\operatorname{o}}):=\eta(\operatorname{L}^{\operatorname{o}})(0)$ where,
\begin{equation}
\eta(\operatorname{L}^{\operatorname{o}})(s):=\sum_{\lambda\in\text{spec}^{*}(\operatorname{L}^{\operatorname{o}})}(\operatorname{sgn}\lambda)|\lambda|^{-s}.
\end{equation}
$\eta(\operatorname{L}^{\operatorname{o}})$ has a rigorous mathematical meaning using the fact that $\eta(\operatorname{L}^{\operatorname{o}})(s)$ admits a meromorphic extension to $\C$ that is regular at $0$, \cite{aps1}.
\begin{rem}
The eta-invariant is an analytic invariant introduced by Atiyah, Patodi and Singer \cite{aps1} defined for an elliptic and self-adjoint operator.  As in \cite[Prop. 4.20]{aps1}, we may remove some spectral symmetry and the eta invariant of $\operatorname{L}^{\operatorname{o}}$ coincides with the eta invariant of the operator $\star d$ restricted to $\Omega^{1}(X,\frak{t})\cap \operatorname{Im}(d\star)$.  Throughout, we will abuse notation slightly and write,
\begin{equation}\label{anometa}
\eta(\star d)=\lim_{s\rightarrow 0}\sum_{\lambda\in\operatorname{spec}^{*}(\star d)}\operatorname{sgn}(\lambda)|\lambda|^{-s},
\end{equation}
and replace $\operatorname{L}^{\operatorname{o}}$ in the notation with $\star d$.
We also recall that the expression for the sum,
\begin{equation}
\sum_{\lambda\in\operatorname{spec}^{*}(\star d)}\operatorname{sgn}(\lambda)|\lambda|^{-s},
\end{equation}
is defined for large $\operatorname{Re}(s)$ and \cite{aps1} shows that it has a meromorphic continuation to $\C$ that is analytic at $0$.  It therefore makes sense to take the limit as $s\rightarrow 0$ in Eq. \eqref{anometa} and to define the eta-invariant $\eta(\star d)$ as evaluation of this limit.
\end{rem}
Let $\eta_{\operatorname{grav}}(\operatorname{g})$ be the eta-invariant for the operator $\star d$ acting on $\Omega^{1}(X,\R)$, so that,
\begin{equation}\label{tetad}
\eta(\star d)=N\cdot\eta_{\operatorname{grav}}(\operatorname{g}),
\end{equation}
where the eta invariant on the left hand side of \eqref{tetad} is defined on $\Omega^{1}(X,\frak{t})$ and $N=\operatorname{dim}\mathbb{T}$.  Since $\eta(\star d)$ itself is not a topological invariant, we follow \cite{w3} and add a ``counterterm'' that cancels the metric dependence of the eta-invariant.
Define,
\begin{equation}\label{gcounter}
\operatorname{CS}_{s}(A^{\operatorname{g}}):=\frac{1}{4\pi}\int_{X}s^{*}\operatorname{Tr}(A^{\operatorname{g}}\wedge dA^{\operatorname{g}}+\frac{2}{3} A^{\operatorname{g}}\wedge A^{\operatorname{g}}\wedge A^{\operatorname{g}}),
\end{equation}
the gravitational Chern-Simons term with $A^{\operatorname{g}}$ the Levi-Civita connection and $s$ a trivializing section of twice the tangent bundle of $X$.  More explicitly, let $H=\operatorname{Spin}(6)$, $Q=TX\oplus TX$ viewed as a principal
$\operatorname{Spin}(6)$-bundle over $X$, $\operatorname{g}\in\Gamma(S^{2}(T^{*}X))$ a Riemannian metric on $X$, $\phi:Q\rightarrow \operatorname{SO}(X)$ a principal bundle morphism, and $A^{LC}\in\mathcal{A}_{SO(X)}:=\{ A\in (\Omega^{1}(\operatorname{SO}(X))\otimes\frak{so}(3))^{\operatorname{SO}(3)}\,\,|\,\,A(\xi^{\sharp})=\xi, \,\,\forall\,\xi\in\frak{so}(3)\}$ the Levi-Civita connection.  Then $A^{\operatorname{g}}:=\phi^{*}A^{LC}\in \mathcal{A}_{Q}:=\{ A\in (\Omega^{1}(Q)\otimes\frak{h})^{H}\,\,|\,\,A(\xi^{\sharp})=\xi, \,\,\forall\,\xi\in\frak{h}\}$.  An Atiyah-Patodi-Singer theorem, \cite[Prop. 4.19]{aps2}, says that the combination,
\begin{equation}\label{apsthm}
\eta_{\operatorname{grav}}(\operatorname{g})+\frac{1}{3}\frac{\operatorname{CS}(A^{\operatorname{g}})}{2\pi},
\end{equation}
is a topological invariant depending only on a $2$-framing of $X$.  Recall that a $2$-framing is a choice of a homotopy equivalence class $\Pi$ of trivializations of $TX\oplus TX$, twice the tangent bundle of $X$.  Note that $\Pi$ is represented by the trivializing section $s:X\rightarrow Q$ above.  The possible $2$-framings correspond to $\Z$.  The identification with $\Z$ is given by the signature defect defined by,
\begin{equation*}
\delta(X,\Pi)=\text{sign}(M)-\frac{1}{6}p_{1}(2TM,\Pi),
\end{equation*}
where $M$ is a $4$-manifold with boundary $X$ and $p_{1}(2TM,\Pi)$ is the relative Pontrjagin number associated to the framing $\Pi$ of the bundle $TX\oplus TX$.  The canonical $2$-framing $\Pi^{c}$ corresponds to $\delta(X,\Pi^{c})=0$.  Thus, overall we replace $\operatorname{sgn}\operatorname{L}^{\operatorname{o}}$ in \eqref{stationdef} with,
\begin{equation}
N\cdot\left[\eta_{\operatorname{grav}}(\operatorname{g})+\frac{1}{3}\frac{\operatorname{CS}(A^{\operatorname{g}})}{2\pi}\right],
\end{equation}
which is a topological invariant up to a choice of two-framing on $X$.\\
\\
Next, we consider the determinant $\operatorname{det}'\operatorname{L}^{\operatorname{o}}$ in \eqref{stationdef}.
Recall the Hodge-de Rham Laplacian,
\begin{equation}\label{harlap}
\Delta_{q}:=d^{\dagger}d+dd^{\dagger},\,\,\,\text{on}\,\,\,\Omega^{q}(X,\frak{t}),
\end{equation}
Let $\zeta_{q}(s)$ denote the zeta function of $\Delta_{q}$,
\begin{equation}
\zeta_{q}(s)=\zeta(\Delta_{q})(s):=\sum_{\lambda\in\text{spec}^{*}(\Delta_{q})}\lambda^{-s}.
\end{equation}
Recall, $\zeta_{q}(s)$ is defined for $\operatorname{Re}(s)\gg0$ by,
\begin{equation}
\zeta_{q}(s):=\frac{1}{\Gamma(s)}\int_{0}^{\infty}t^{s-1}\operatorname{tr}(e^{t\Delta_{q}}-\Pi_{q})dt,
\end{equation}
and then analytically continued to $\C$ as usual.  Note that $\Pi_{q}:\Omega^{q}(\operatorname{M},\rho)\rightarrow \mathcal{H}^{q}(\operatorname{M},\rho)$ is orthogonal projection, and $\Gamma(s)$ is the gamma function,
\begin{equation*}
\Gamma(s)=\int_{0}^{\infty}t^{s-1}e^{-t}dt.
\end{equation*}
The notation $\operatorname{det}'$ refers to a regularized determinant and is defined for the Laplacians $\Delta_{q}$ as,
\begin{equation}\label{regdet1}
\operatorname{det}'(\Delta_{q}):=e^{-\zeta'(\Delta_{q})(0)}.
\end{equation}
The scaling used in \eqref{scalingeq} is a consequence of the following,
\begin{lemma}\label{scallem}
For any $c\in\R_{+}$,
\begin{eqnarray*}\label{forwant}
\operatorname{det}'\left[c\cdot \Delta_{q}\right]&=&c^{\zeta_{q}(0)}\cdot\operatorname{det}'\left[\Delta_{q}\right],\\\label{forwant2}
                                                 &=&c^{-\operatorname{dim}H^{q}(X,\frak{t})}\cdot\operatorname{det}'\left[\Delta_{q}\right].
\end{eqnarray*}
\end{lemma}
\begin{proof}
By definition, $\zeta(c\Delta_{q})(s)=c^{-s}\zeta(\Delta_{q})(s)$.  Taking the derivative of $c^{-s}\zeta(\Delta_{q})(s)$ with respect to $s$ and evaluating at $s=0$ and using the definition \eqref{regdet1} yields \eqref{forwant}.  In order to obtain the precise scaling in \eqref{forwant2} we use the following \cite{mul},
\begin{equation*}
\zeta_{q}(0)=-\operatorname{dim}\operatorname{Ker}\Delta_{q}=-\operatorname{dim}H^{q}(X,\frak{t}).
\end{equation*}
This completes the proof.
\end{proof}
Now we define the determinant $\operatorname{det}'\operatorname{L}^{\operatorname{o}}$ as,
\begin{eqnarray*}\label{detdef}
\operatorname{det}'\operatorname{L}^{\operatorname{o}}&:=&\left[\operatorname{det}'(\operatorname{L}^{\operatorname{o}})^{2}\right]^{1/2},\\
                                                      &=&\left[\operatorname{det}'(\Delta_{1}\oplus\Delta_{3})\right]^{1/2},\\
                                                      &=&\left[\operatorname{det}'\Delta_{1}\right]^{1/2}\cdot\left[\operatorname{det}'\Delta_{3}\right]^{1/2}.
\end{eqnarray*}
Note that $\Delta_{1}\oplus\Delta_{3}$ denotes the operator acting on $\Omega^{1}(X,\frak{t})\oplus \Omega^{3}(X,\frak{t})$ in the obvious way, preserving the direct sum.  The quantity of interest in equation \eqref{stationdef} is,
\begin{eqnarray*}
\frac{\operatorname{det}'[d^{\dagger}d]}{\sqrt{|\operatorname{det}'[k\operatorname{L}^{\operatorname{o}}]|}}&=&\frac{\operatorname{det}'\Delta_{0}}{\sqrt{|\left[\operatorname{det}'k^{2}\Delta_{1}\right]^{1/2}\cdot\left[\operatorname{det}'k^{2}\Delta_{3}\right]^{1/2}|}},\\
                                                                                                            &=&k^{\frac{1}{2}\left(\operatorname{dim}H^{1}(X,\frak{t})+\operatorname{dim}H^{0}(X,\frak{t})\right)}\cdot\frac{\left[\operatorname{det}'\Delta_{0}\right]^{3/4}}{\left[\operatorname{det}'\Delta_{1}\right]^{1/4}},
\end{eqnarray*}
where the last line follows from Lemma \ref{scallem} and the fact that $\operatorname{det}'\Delta_{0}=\operatorname{det}'\Delta_{3}$ by duality.  Overall, we obtain the following for $Z_{\mathbb{T}}(X,P,k)$,
\begin{eqnarray}
&&K'(A_{P},k)\cdot\int_{\mathcal{M}_{P}}\,\,\frac{1}{\sqrt{|\operatorname{det}'k\operatorname{L}^{\operatorname{o}}|}}\operatorname{exp}\left[\frac{i\pi}{4}\operatorname{sgn}(\operatorname{L}^{\operatorname{o}})\right]\operatorname{det}'\left[d^{\dagger}d\right]\,\nu,\\\label{almostfin}
                     &:=&k^{m_{X}}e^{i k \operatorname{CS}_{X,P}(A_{P})}e^{\left[\frac{i\pi N}{4}\cdot\left[\eta_{\operatorname{grav}}(\operatorname{g})+\frac{1}{3}\frac{\operatorname{CS}(A^{\operatorname{g}})}{2\pi}\right]\right]}\int_{\mathcal{M}_{P}}\,\,\frac{1}{\operatorname{Vol}I}\frac{\left[\operatorname{det}'\Delta_{0}\right]^{3/4}}{\left[\operatorname{det}'\Delta_{1}\right]^{1/4}}\,\nu,
\end{eqnarray}
where $m_{X}:=\frac{1}{2}\left(\operatorname{dim}H^{1}(X,\frak{t})-2\operatorname{dim}H^{0}(X,\frak{t})\right)$.
\begin{rem}
Note that the term $m_{X}:=\frac{1}{2}\left(\operatorname{dim}H^{1}(X,\frak{t})-2\operatorname{dim}H^{0}(X,\frak{t})\right)$ results in a difference in the $k$-dependence of our partition function from that of \cite{m} by a factor of $k^{-\frac{1}{2}\operatorname{dim}H^{0}(X,\frak{t})}$.  We also note that we obtain the same $k$-dependent term $k^{m_{X}}$ in the \emph{symplectic abelian Chern-Simons partition function} given in definition \ref{symrigdef} using the completely different technique of non-abelian localization.
\end{rem}
Next, we will show that the quantity inside the integral in \eqref{almostfin},
$$\frac{1}{\operatorname{Vol}I}\frac{\left[\operatorname{det}'\Delta_{0}\right]^{3/4}}{\left[\operatorname{det}'\Delta_{1}\right]^{1/4}}\,\nu,$$
is precisely the square-root of the Reidemiester-Ray-Singer torsion of $X$.  The Reidemeister-Ray-Singer torsion $T_{X}$ will be defined as a density on the determinant line,
\begin{equation}
|\operatorname{det}H^{\bullet}(X,\frak{t})|^{*}:=\bigotimes_{j=0}^{3}(\operatorname{det}H^{j}(X,\frak{t}))^{(-1)^{j}}.
\end{equation}
We make the natural identification,
\begin{equation}
H^{\bullet}(X,\frak{t})\simeq \mathcal{H}^{\bullet}(X,\frak{t}),
\end{equation}
under the de Rham map where $\mathcal{H}^{\bullet}(X,\frak{t})$ denotes the harmonic forms on $X$ with respect to the Laplacian given in Eq. \eqref{harlap}.
Let $\delta_{|\operatorname{det}H^{\bullet}(X,\frak{t})|^{*}}$ denote the induced density on $|\operatorname{det}H^{\bullet}(X,\frak{t})|^{*}$ corresponding to the induced metric from the $L^{2}$ metric on $\mathcal{H}^{\bullet}(X,\frak{t})$.
Now make the following,
\begin{define}\cite{rsi}\label{scalartors}
Given a closed Riemannian three manifold $(X,\operatorname{g})$, define the \emph{scalar} Reidemeister-Ray-Singer torsion,
\begin{equation}\label{rsdef1}
T_{X}^{\operatorname{scal}}(\operatorname{g}):=\operatorname{exp}\left( \frac{1}{2}\sum_{q=0}^{3}(-1)^{q}q\zeta'(\Delta_{q})(0)\right).
\end{equation}
Define the Reidemeister-Ray-Singer torsion $T_{X}$ as,
\begin{equation}\label{rsdef2}
T_{X}:=T_{X}^{\operatorname{scal}}(\operatorname{g})\cdot \delta_{|\operatorname{det}H^{\bullet}(X,\frak{t})|^{*}}.
\end{equation}
\end{define}
Note that $T_{X}^{\operatorname{scal}}(\operatorname{g})$ is generally dependent upon the choice of metric $\operatorname{g}$ and it is shown in \cite{rsi} that $T_X$ is indeed independent of $\operatorname{g}$.  Note that given an orthonormal basis for $\mathcal{H}^{q}(X,\R)$, $\left\{\nu_{1}^{[q]},\ldots,\nu_{b_{q}}^{[q]}\right\}$, where $b_{q}:=\operatorname{dim}H^{1}(X,\R)$, $\delta_{|\operatorname{det}H^{\bullet}(X,\frak{t})|^{*}}$ may be written as,
\begin{equation}\label{orthoexp}
\delta_{|\operatorname{det}H^{\bullet}(X,\frak{t})|^{*}}=\bigotimes_{q=0}^{3}\left|\nu^{[q]}\right|^{N\cdot(-1)^{q}},
\end{equation}
where $\nu^{[q]}:=\nu_{1}^{[q]}\wedge\cdots\wedge\nu_{b_{q}}^{[q]}$ and $N=\operatorname{dim}\mathbb{T}$.  Observe that an orthonormal basis for $\mathcal{H}^{q}(X,\R)=\R$ is a constant $\nu^{[0]}$ such that $\left|\nu^{[0]}\right|^{N}=\left(\operatorname{Vol}I\right)^{-1}$.  One may see this by computing,
\begin{eqnarray*}
1&=&||\nu^{[0]}||^{2}_{L^{2}},\\
 &=&\int_{X}\nu^{[0]}\wedge \star \nu^{[0]},\\
 &=&|\nu^{[0]}|^{2}\int_{X}\star 1.
\end{eqnarray*}
Combining this with \eqref{volres} one obtains $\left|\nu^{[0]}\right|^{N}=\left(\operatorname{Vol}I\right)^{-1}$.  Using Poincar\'{e} duality $H^{q}(X,\R)\simeq H^{3-q}(X,\R)^{*}$ we may write the square-root of $\delta_{|\operatorname{det}H^{\bullet}(X,\frak{t})|^{*}}$ in \eqref{orthoexp} as,
\begin{eqnarray}\nonumber
\left[\delta_{|\operatorname{det}H^{\bullet}(X,\frak{t})|^{*}}\right]^{1/2}&=&\left|\nu^{[0]}\right|^{N}\otimes\left|\nu^{[1]}\right|^{N},\\\label{lastnu}
                                                                           &=&\frac{1}{\operatorname{Vol}I}\cdot\nu,
\end{eqnarray}
where we define $\nu:=\left|\nu^{[1]}\right|^{N}$.  Thus, using Poincar\'{e} duality combined with the duality $\Delta_{q}\simeq \Delta_{3-q}$ induced by the Hodge star, then the definition of $T_{X}$ in \eqref{rsdef1}, \eqref{rsdef2} and equation \eqref{lastnu} imply that the square-root of the Reidemeister-Ray-Singer torsion can be expressed as,
\begin{equation}
\sqrt{T_{X}}=\frac{1}{\operatorname{Vol}I}\frac{\left[\operatorname{det}'\Delta_{0}\right]^{3/4}}{\left[\operatorname{det}'\Delta_{1}\right]^{1/4}}\,\nu,
\end{equation}
as claimed.
Overall, we make the following,
\begin{define}\label{rigpartdef}
Let $k\in\Z$ and $X$ a closed, oriented three-manifold.  The abelian Chern-Simons partition function, $Z_{\mathbb{T}}(X,k)$, is the quantity,
\begin{equation}
Z_{\mathbb{T}}(X,k)=\sum_{P\in\operatorname{Tors}H^{2}(X,\Lambda)}Z_{\mathbb{T}}(X,P,k),
\end{equation}
and,
\begin{equation}
Z_{\mathbb{T}}(X,P,k):=k^{m_X}e^{i k \operatorname{CS}_{X,P}(A_{P})}e^{\pi i N\left(\frac{\eta_{\operatorname{grav}}(\operatorname{g})}{4}+\frac{1}{12}\frac{\operatorname{CS}(A^{\operatorname{g}})}{2\pi}\right)}\int_{\mathcal{M}_{P}}\sqrt{T_{X}},
\end{equation}
where $m_X=\frac{N}{2}(\operatorname{dim}H^{1}(X,\R)-2\operatorname{dim}H^{0}(X,\R))$.
\end{define}
Note that we can either choose the canonical framing \cite{at} and work with this throughout, or we can observe that if the framing of $X$ is twisted by $F$ units, then $\operatorname{CS}(A^{\operatorname{g}})$ transforms by,
\begin{equation*}
\operatorname{CS}(A^{\operatorname{g}})\rightarrow \operatorname{CS}(A^{\operatorname{g}})+2\pi F.
\end{equation*}
The partition function $Z_{\mathbb{T}}(X,k)$ is then transformed by,
\begin{equation}\label{partform}
Z_{\mathbb{T}}(X,k)\rightarrow Z_{\mathbb{T}}(X,k)\cdot \operatorname{exp}\left(\frac{2\pi i N F}{24}\right).
\end{equation}
Thus, $Z_{\mathbb{T}}(X,k)$ is a topological invariant of framed, oriented three-manifolds, with a transformation law under change of framing.  This is tantamount to a topological invariant of oriented three-manifolds without a choice of framing.

\section{Shift Symmetry and the Abelian Partition Function}\label{shiftsec}
Our goal in this section is to obtain a heuristic ``shift invariant'' expression for the abelian Chern-Simons partition function by decoupling one of the three components of the gauge field $A\in \mathcal{A}_{P}$ using a particular symmetry of our theory.  The symmetry that we exhibit is directly associated to a choice of contact structure on our three manifold $X$ and is called a ``shift symmetry.''  In this section, we will assume that $(X,H)$ is a closed contact three manifold and $\kappa\in \Omega^{1}(X)$ is a contact one form, so that $\operatorname{Ker}(\kappa)=H$.  Note that every closed, orientable three manifold admits a contact structure \cite{ma} and therefore the shift symmetry construction is a general symmetry that applies to any closed, orientable three-manifold.  We note that the constructions in this section are largely heuristic and should be viewed as an \emph{initial} step in obtaining a rigorous definition for a shift invariant expression of the abelian Chern-Simons partition function.  First we make the following,
\begin{define}
The action of the space of local shift symmetries $\mathcal{S}$ on $\mathcal{A}_{P}$ is defined by its variation $\delta_{\sigma}$ on a field $A\in\mathcal{A}_{P}$ by,
\begin{equation*}
\delta_{\sigma} A:=\sigma \kappa,
\end{equation*}
where $\sigma\in \Omega^{0}(X,\frak{t})$ is an arbitrary form and $\kappa\in\Omega^{1}(X)$ is a fixed contact form on $X$.
\end{define}
Note that since the shift symmetry is defined for arbitrary $\sigma\in \Omega^{0}(X,\frak{t})$, it is independent of the choice of $\kappa$ for the contact structure $H\subset TX$ since any two such contact forms must be related by a multiple of a non-vanishing scalar function on $X$.  Clearly, the Chern-Simons action, $\operatorname{CS}_{X,P}(A)$, does \emph{not} respect the shift symmetry.  That is,
\begin{equation}
\delta_{\sigma}\operatorname{CS}_{X,P}(A)\neq 0,
\end{equation}
for arbitrary $\sigma\in\Omega^{0}(X,\frak{t})$.
In order to study a shift invariant version of abelian Chern-Simons theory, we follow \cite[\S3.1]{bw} and introduce a new scalar field $\Phi\in \Omega^{0}(X,\frak{t})$ such that,
\begin{equation*}
\delta_{\sigma} \Phi=\sigma.
\end{equation*}
We postulate the scaling,
\begin{equation*}
\Phi\rightarrow t^{-1}\Phi,
\end{equation*}
for a non-zero function $t\in C^{\infty}(X)$ whenever,
\begin{equation*}
\kappa\rightarrow t\kappa,
\end{equation*}
so that $\kappa\Phi\in \Omega^{1}(X,\frak{t})$ is invariant under the scaling by $t$ and is a well defined form, independent of the choice of $\kappa$.  Then for any principal $\mathbb{T}$-bundle $P$ we define a new action,
\begin{eqnarray}
\operatorname{CS}_{X,P}(A,\Phi)&:=&\operatorname{CS}_{X,P}(A-\kappa\Phi)\nonumber\\
               &:=&\int_{X}\alpha(\widehat{A-\kappa\Phi}),\nonumber\\\label{tecres}
                     &=&\int_{X}\alpha(\hat{A}-\kappa\hat{\Phi}),\\\label{lastshift}
                     &=&\operatorname{CS}_{X,P}(A)-\frac{1}{4\pi}\int_{X}[2\kappa\wedge \operatorname{Tr} (\hat{\Phi}\wedge F_{\hat{A}})-\kappa\wedge d\kappa\,\, \operatorname{Tr} (\hat{\Phi}^{2})],
\end{eqnarray}
where Eq. \eqref{tecres} follows from the definition of $\hat{A}$ and $\hat{\Phi}$, where, $$\hat{\Phi}|_{[p,h]}:=\operatorname{Ad}_{h^{-1}}(\iota_{*}\operatorname{pr}_{1}^{*}\Phi|_{h}),$$
on $\hat{P}=P\times_{\mathbb{T}} SU(N+1)$.  It is easy to see that the new action, $\operatorname{CS}_{X,P}(A,\Phi)$, is invariant under the shift symmetry.  Now define a ``new'' partition function,
\begin{equation}\label{newpart1}
\bar{Z}_{\mathbb{T}}(X,P,k):=\frac{1}{\operatorname{Vol}(\mathcal{S})}\frac{1}{\operatorname{Vol}(\mathcal{G})}\left(\frac{k}{4\pi^2}\right)^{\Delta{\mathcal{G}}} \ddashint_{\mathcal{A}_{P}}\mathcal{D}A\,\mathcal{D}\Phi\,\, e^{i k \operatorname{CS}_{X,P}(A,\Phi)},
\end{equation}
where $\mathcal{D}\Phi$ is defined by the invariant, positive definite quadratic form, \cite[Eq. 3.8]{bw},
\begin{equation}\label{phiprod}
(\Phi,\Phi)=-\int_{X}\langle\Phi,\Phi\rangle\,\kappa\wedge d\kappa.
\end{equation}
As observed in \cite{bw}, the new partition function of Eq. \eqref{newpart1} should be identically equal to the original partition function defined for abelian Chern-Simons theory as in Eq. \eqref{hpart2},
\begin{equation}
Z_{\mathbb{T}}(X,P,k)=\frac{1}{\operatorname{Vol}(\mathcal{G})}\left(\frac{k}{4\pi^2}\right)^{\Delta{\mathcal{G}}} \ddashint_{\mathcal{A}_{P}} \mathcal{D}A\, e^{i k \operatorname{CS}_{X,P}(A)}.
\end{equation}
This is seen by fixing $\Phi=0$ using the shift symmetry, $\delta_{\sigma}\Phi=\sigma$, which will cancel the pre-factor $\operatorname{Vol}(\mathcal{S})$ from the resulting group integral over $\mathcal{S}$ and yield exactly our original partition function,
\begin{equation*}
Z_{\mathbb{T}}(X,P,k)=\frac{1}{\operatorname{Vol}(\mathcal{G})}\left(\frac{k}{4\pi^2}\right)^{\Delta{\mathcal{G}}} \ddashint_{\mathcal{A}_{P}}\mathcal{D}A\,\, e^{i k \operatorname{CS}_{X,P}(A)}.
\end{equation*}
Thus, we obtain the heuristic result,
\begin{equation}\label{equivpart}
\bar{Z}_{\mathbb{T}}(X,k)=Z_{\mathbb{T}}(X,k).
\end{equation}
On the other hand, we obtain another description of $\bar{Z}_{\mathbb{T}}(X,P,k)$ by integrating $\Phi$ out.  Our new description of the partition function is,
\begin{equation}\label{part}
\bar{Z}_{\mathbb{T}}(X,P,k)=C\cdot\ddashint_{\mathcal{A}_{P}}\mathcal{D}A\,\, \exp\,\left[i k\left(\operatorname{CS}_{X,P}(A)-\frac{1}{4\pi}\int_{X} \frac{\operatorname{Tr}[(\kappa\wedge F_{\hat{A}})^{2}]}{\kappa\wedge d\kappa}\right)\right],
\end{equation}
where $C=\frac{1}{\operatorname{Vol}(\mathcal{S})}\frac{1}{\operatorname{Vol}(\mathcal{G})}\left(\frac{-ik}{4\pi^2}\right)^{\Delta{\mathcal{G}/2}}$.  We rewrite this partition function after identifying $\mathcal{A}_{P}=A_{P}+\Omega^{1}(X,\frak{t})$ for a flat base point $A_{P}$ in $\mathcal{A}_{P}$.  We then obtain,
\begin{equation}\label{Anom1a}
\bar{Z}_{\mathbb{T}}(X,P,k)=C_1\cdot\ddashint_{\mathcal{A}_{P}}\mathcal{D}A\,\, \exp\,\left[\frac{i k}{4\pi}\left(\int_{X} \langle A\wedge dA\rangle-\int_{X} \frac{\langle(\kappa\wedge dA)^{2}\rangle}{\kappa\wedge d\kappa}\right)\right],
\end{equation}
where,
$$C_1 = \frac{e^{i k \operatorname{CS}_{X,P}(A_{P})}}{\operatorname{Vol}(\mathcal{S})\operatorname{Vol}(\mathcal{G})} \left(\frac{-ik}{4\pi^2}\right)^{\Delta{\mathcal{G}/2}}.$$
\noindent
Note that the critical points of this action, up to the action of the shift symmetry, are precisely the flat connections, \cite[Eq. 5.3]{bw}.  We abuse notation and write $A\in T_{A_{P}}\mathcal{A}_{P}$.  Let us define the notation,
\begin{equation}\label{newact}
4\pi\overline{\operatorname{CS}}(A):=\int_{X} \langle A\wedge dA\rangle-\int_{X} \frac{\langle(\kappa\wedge dA)^{2}\rangle}{\kappa\wedge d\kappa},
\end{equation}
for the new action that appears in the partition function.  Also define,
\begin{equation}
\overline{\operatorname{S}}(A):=\frac{1}{4\pi}\int_{X} \frac{\langle(\kappa\wedge dA)^{2}\rangle}{\kappa\wedge d\kappa},
\end{equation}
so that we may write,
\begin{equation}
\overline{\operatorname{CS}}(A)=\operatorname{CS}(A)-\overline{\operatorname{S}}(A).
\end{equation}
The primary virtue of Eq. \eqref{Anom1a} above is that it is heuristically equal to the original Chern-Simons partition function of Def. \eqref{heurdef} and yet it is expressed in such a way that the action $\overline{\operatorname{CS}}(A)$ is invariant under the shift symmetry.  This means that $\overline{\operatorname{CS}}(A+\sigma\kappa)=\overline{\operatorname{CS}}(A)$ for all tangent vectors $A\in T_{A_{P}}(\mathcal{A}_{P})\simeq\Omega^{1}(X,\frak{t})$ and $\sigma\in \Omega^{0}(X,\frak{t})$.  We may naturally view $A\in\Omega^1(H,\frak{t})$, the subset of $\Omega^{1}(X,\frak{t})$ restricted to the contact distribution $H\subset TX$.  If $\xi$ denotes the Reeb vector field of $\kappa$, then $\Omega^{1}(H,\frak{t})=\{\alpha\in\Omega^{1}(X,\frak{t})\,\,|\,\,\iota_{\xi}\alpha=0\}$.  The remaining contributions to the partition function come from the orbits of $\mathcal{S}$ in $\mathcal{A}_{P}$, which turn out to give a contributing factor of $\operatorname{Vol}(\mathcal{S})$, \cite[Eq. 3.32]{bw}.  We thus reduce our integral to an integral over $\bar{\mathcal{A}}_{P}:=\mathcal{A}_{P}/\mathcal{S}$ and obtain,
\begin{eqnarray*}
Z_{\mathbb{T}}(X,P,k)&=&\frac{e^{i k \operatorname{CS}_{X,P}(A_{P})}}{\operatorname{Vol}(\mathcal{G})}\ddashint_{\bar{\mathcal{A}}_{P}}\bar{\mathcal{D}}A\,\, \exp\,\left[\frac{i k}{4\pi}\left(\int_{X}\langle A\wedge dA\rangle-\int_{X} \frac{\langle(\kappa\wedge dA)^{2}\rangle}{\kappa\wedge d\kappa}\right)\right],\\
               &=&\frac{e^{i k \operatorname{CS}_{X,P}(A_{P})}}{\operatorname{Vol}(\mathcal{G})}\ddashint_{\bar{\mathcal{A}}_{P}}\bar{\mathcal{D}}A\,\, \exp\,\left[i k\overline{\operatorname{CS}}(A)\right],
\end{eqnarray*}\\
where $\bar{\mathcal{D}}A$ denotes an appropriate quotient measure on $\bar{\mathcal{A}}_{P}$, and $A\in\Omega^1(H,\frak{t})\simeq T_{A_{P}}\bar{\mathcal{A}}_{P}$.
\begin{rem}\label{shiftdef1}
We now make a new heuristic definition of a partition function.  Let $k\in\Z$ and $(X,\kappa)$ a closed, oriented contact three-manifold.  The \emph{shift reduced abelian Chern-Simons partition function, $\bar{Z}_{\mathbb{T}}(X,k)$,} is the heuristic quantity,
\begin{equation}
\bar{Z}_{\mathbb{T}}(X,k)=\sum_{[P]\in\operatorname{Tors}H^{2}(X,\Lambda)}\bar{Z}_{\mathbb{T}}(X,P,k),
\end{equation}
and,
\begin{equation}\label{Anom12}
\bar{Z}_{\mathbb{T}}(X,P,k)=\frac{e^{i k \operatorname{CS}_{X,P}(A_{P})}}{\operatorname{Vol}(\mathcal{G})}\left(\frac{-ik}{4\pi^2}\right)^{\Delta{\mathcal{G}/2}}\ddashint_{\bar{\mathcal{A}}_{P}}\bar{\mathcal{D}}A\,\, \exp\,\left[i k\overline{\operatorname{CS}}(A)\right],
\end{equation}
and $4\pi\overline{\operatorname{CS}}(A):=\int_{X} \langle A\wedge dA\rangle-\int_{X} \frac{\langle(\kappa\wedge dA)^{2}\rangle}{\kappa\wedge d\kappa}$ is the \emph{shift reduced Chern-Simons action}.
\end{rem}

\section{Moment Map Squared Form of the Partition Function}\label{momsquared}
Our starting point is a heuristically defined partition function defined in remark \ref{shiftdef1}.
Our next objective is to determine a moment map $\mu$ for a group action such that the shift reduced action may be identified as the moment map squared, $\overline{\operatorname{CS}}(A)=(\mu,\mu)$.  We note that there may be a simplification of our considerations in this section using the fact that $\mathbb{T}$ is abelian.  We proceed as in the general case, however.  First we observe that the gauge group $\mathcal{G}$ itself cannot act in a Hamiltonian fashion such that $\overline{\operatorname{CS}}(A)=(\mu,\mu)$ since $\overline{\operatorname{CS}}(A)$ is not invariant under ``large'' gauge transformations.  We therefore restrict to the connected component of the gauge group containing the identity element, $\mathcal{G}_{0}$.  Let $Y\in\operatorname{Lie}(\mathcal{G}_0)\simeq \Omega^{0}(X,\frak{t})$, and let $Y^{\#}:=d^{H}Y\in\Gamma(T\bar{\mathcal{A}})$ denote the vector field generated by $Y$ on $\bar{\mathcal{A}}$.  By definition of the symplectic form $\Omega$ on $\bar{\mathcal{A}}$, we have,
$$(\iota_{Y^{\#}}\Omega)(\delta A)=-\int_{X}\kappa\wedge\langle dY\wedge \delta A\rangle.$$
If,
$$\mu^{Y}(A):=\int_{X}\kappa\wedge\langle Y\wedge F_{A}\rangle-\int_{X}d\kappa\wedge\langle Y\wedge A\rangle,$$
then one can show that the moment map equation is satisfied,
$$d\mu^{Y}=\iota_{Y^{\#}}\Omega,$$
using integration by parts.  We observe that the $\mathcal{G}_{0}$ action on $\bar{\mathcal{A}}$ is not Hamiltonian, however.  This may be checked by computing the Poisson bracket,
\begin{eqnarray*}
\left\{\mu^{Y_{1}},\mu^{Y_{2}}\right\}&=&\Omega(d^{H}Y_{1},d^{H}Y_{2}),\\
                                     &=&-\int_{X}\kappa\wedge\langle dY_{1}\wedge dY_{2}\rangle,\\
                                     &=&\int_{X}\kappa\wedge\langle[Y_{1},Y_{2}]\wedge F_{A}\rangle-\int_{X}d\kappa\wedge\langle Y_{1}\wedge dY_{2}\rangle,\\
                                     &=&\mu^{[Y_{1},Y_{2}]}-\int_{X}d\kappa\wedge\langle Y_{1}\wedge dY_{2}\rangle.
\end{eqnarray*}
The obstruction for the map $\mu$ to determine a moment map is given by the cocycle,
\begin{eqnarray}\label{cocycle}
c(Y_{1},Y_{2})&:=&\left\{\mu^{Y_{1}},\mu^{Y_{2}}\right\}-\mu^{[Y_{1},Y_{2}]},\\\nonumber
               &=&-\int_{X}d\kappa\wedge\langle Y_{1}\wedge dY_{2}\rangle,\\\nonumber
               &=&-\int_{X}\langle Y_{1}\wedge \mathcal{L}_{\xi}Y_{2}\rangle\,\, \kappa\wedge d\kappa,
\end{eqnarray}
which clearly does not vanish in general.  Following the ideas of \cite[\S 3.4]{bw}, one may obtain a Hamiltonian action by considering a central extension $\widetilde{\mathcal{G}}_{0}$ by $\operatorname{U}(1)$ of the group $\mathcal{G}_{0}$ determined by the cocycle $c(Y_{1},Y_{2})$,
$$\operatorname{U}(1)\rightarrow \widetilde{\mathcal{G}}_{0}\rightarrow \mathcal{G}_{0}.$$
As in \cite{bw}, we assume that the central $\operatorname{U}(1)$ subgroup of $\widetilde{\mathcal{G}}_{0}$ acts trivially on $\bar{\mathcal{A}}$ and the moment map for the central generator $(0,a)$ of the Lie algebra is constant.  We then see that the new moment map for the action of $\widetilde{\mathcal{G}}_{0}$ on $\bar{\mathcal{A}}$,
$$\mu^{(Y,a)}(A):=\int_{X}\kappa \wedge \langle Y \wedge F_{A}\rangle - \int_{X}d\kappa\langle Y\wedge A\rangle + a,$$
is Hamiltonian,
$$\left\{\mu^{(Y_{1},a_{1})},\mu^{(Y_{2},a_{2})}\right\}=\mu^{[(Y_{1},a_{1}),(Y_{2},a_{2})]},$$
where,
$$[(Y_{1},a_{1}),(Y_{2},a_{2})]:=\left([Y_{1},Y_{2}],c(Y_{1},Y_{2})\right).$$
In order to cast the action $\overline{\operatorname{CS}}$ into a moment map squared form the Lie algebra of the Hamiltonian group must admit a non-degenerate invariant inner product.  The group $\widetilde{\mathcal{G}}_{0}$ does \emph{not} admit such an inner product, however.  
\begin{rem}
Following \cite{bw}, this problem may be solved by assuming that $X$ admits a Seifert structure.  We have made no assumption about the contact structure up until this point and all of our considerations have been completely valid for the general case.  We now make the assumption that the contact structure is such that the Reeb vector field generates a locally free $\U(1)$ action such that the first orbifold Chern number $c_{1}(X)\neq 0$.  It turns out that the closed three manifolds $X$ that admit a locally free $\U(1)$ action such that $c_{1}(X)\neq 0$ are precisely the (quasi-regular) Sasakian manifolds \cite[Theorem 7.5.2]{bg}.
\end{rem}
Since the action of $\U(1)$ on $X$ induces an action on $\widetilde{\mathcal{G}}_{0}$, we naturally consider the semidirect product $\U(1)\ltimes \widetilde{\mathcal{G}}_{0}$, which admits the non-degenerate invariant inner product on the Lie algebra of $\U(1)\ltimes \widetilde{\mathcal{G}}_{0}$,
\begin{equation}\label{pairing1}
\left((p_{1},Y_{1},a_{1}),(p_{2},Y_{2},a_{2})\right)=\left(-\int_{X}\langle Y_{1}\wedge Y_{2}\rangle\,\,\kappa\wedge d\kappa\right) - p_{1}a_{2}-p_{2}a_{1}.
\end{equation}
The Lie bracket on $\operatorname{Lie}(\U(1)\ltimes \widetilde{\mathcal{G}}_{0})$ is also given by,
\begin{equation}\label{lieloop}
\left[(p_{1},Y_{1},a_{1}),(p_{2},Y_{2},a_{2})\right]:=(0,[Y_{1},Y_{2}]+p_{1}\mathcal{L}_{\xi}Y_{2}-p_{2}\mathcal{L}_{\xi}Y_{1},c(Y_{1},Y_{2})).
\end{equation}
One can show that the vector field on $\mathcal{A}$ generated by an element $\textbf{Y}=(p,Y,a)\in\operatorname{Lie}\left(\U(1)\ltimes \widetilde{\mathcal{G}}_{0}\right)$ is given by,
\begin{equation}\label{hvect}
\textbf{Y}^{\#}(A)=dY+p\mathcal{L}_{\xi}A.
\end{equation}
Using \eqref{hvect}, one can see that the moment map in direction of the generator $\textbf{Y}_{p}:=(p,0,0)$ may be given by,
$$\mu^{\textbf{Y}_{p}}(A)=-\frac{1}{2}p\int_{X}\kappa\wedge\langle \mathcal{L}_{\xi}A\wedge A\rangle.$$
Clearly, $\mu^{\textbf{Y}_{p}}(A)$ is invariant under the shift symmetry and descends to $\bar{\mathcal{A}}$.  We claim that the action of $\U(1)\ltimes \widetilde{\mathcal{G}}_{0}$ on $\bar{\mathcal{A}}$ defined above is Hamiltonian with moment map,
\begin{equation}\label{momenteq1}
\mu^{\textbf{Y}}(A)=-\frac{1}{2}p\int_{X}\kappa\wedge\langle \mathcal{L}_{\xi}A\wedge A\rangle+\int_{X}\kappa\wedge \langle Y \wedge dA\rangle-\int_{X}d\kappa\wedge \langle Y \wedge A\rangle+a.
\end{equation}
Let $\textbf{Y}_{Y}:=(0,Y,0)\in \operatorname{Lie}\left(\U(1)\ltimes \widetilde{\mathcal{G}}_{0}\right)$.  In order to show that \eqref{momenteq1} is a moment map we need only compute,
$$\left\{\mu^{\textbf{Y}_{p}},\mu^{\textbf{Y}_{Y}}\right\},$$
since this is the only non-trivial Poisson bracket that remains to be checked.  We compute,
\begin{eqnarray*}
\left\{\mu^{\textbf{Y}_{p}},\mu^{\textbf{Y}_{Y}}\right\}&=&\Omega(p\mathcal{L}_{\xi}A,dY),\\
                                                        &=&-p\int_{X}\kappa\wedge \langle \mathcal{L}_{\xi}A\wedge dY\rangle,\\
                                                        &=&p\int_{X}\kappa\wedge \langle \mathcal{L}_{\xi}Y\wedge dA\rangle-p\int_{X}d\kappa\wedge \langle \mathcal{L}_{\xi}Y\wedge A\rangle,\\
                                                        &=&\mu^{\mathbf{Y}_{p\mathcal{L}_{\xi}Y}},
\end{eqnarray*}
where $\mathbf{Y}_{p\mathcal{L}_{\xi}Y}:=(0,p\mathcal{L}_{\xi}Y,0)$.  By definition of the Lie bracket in \eqref{lieloop}, our last computation shows that the moment map condition is satisfied.  We therefore take $\mathcal{H}:=\U(1)\ltimes \widetilde{\mathcal{G}}_{0}$ to be the Hamiltonian group for abelian Chern-Simons theory.\\
\\
We now claim that the action $\overline{\operatorname{CS}}$ may be expressed in moment map squared form $(\mu,\mu)$ for the moment map defined in \eqref{momenteq1}.  Let $\frak{H}:=\operatorname{Lie}(\mathcal{H})$ and $\langle\langle\cdot,\cdot\rangle\rangle:\frak{H}^{*}\otimes\frak{H}\rightarrow \R$ denote the dual pairing on $\frak{H}$.  Note that we have implictly been using the notation $\mu^{\mathbf{Y}}$ to mean,
$$\mu^{\mathbf{Y}}=\langle\langle\mu,\mathbf{Y}\rangle\rangle,$$
where we view $\mu\in\frak{H}^{*}$.  We make the identification $\frak{H}^{*}\simeq \frak{H}$ via the pairing $(\cdot,\cdot)$ defined in \eqref{pairing1}.  Let $\hat{\mu}\in\frak{H}$ be defined by $\mu=(\hat{\mu},\cdot)$.  We abuse notation and write,
$$(\mu,\mu):=(\hat{\mu},\hat{\mu}).$$
One can check that,
\begin{equation}
\hat{\mu}=\left(-1,-\left(\frac{\kappa\wedge dA-d\kappa\wedge A}{\kappa\wedge d\kappa}\right),\frac{1}{2}\int_{X}\kappa\wedge\langle\mathcal{L}_{\xi}A\wedge A\rangle \right),
\end{equation}
and indeed,
\begin{equation}
\mu^{\mathbf{Y}}=\langle\langle\mu,\mathbf{Y}\rangle\rangle=\left(\hat{\mu},\mathbf{Y} \right),
\end{equation}
for all $\mathbf{Y}\in \frak{H}$.  Thus, by definition,
\begin{eqnarray}\label{momsq}
(\mu,\mu)&=&(\hat{\mu},\hat{\mu})=\langle\langle\mu,\hat{\mu}\rangle\rangle=\mu^{\hat{\mu}},\nonumber\\
         &=&\int_{X}\kappa\wedge\langle\mathcal{L}_{\xi}A\wedge A\rangle-\int_{X}\kappa\wedge d\kappa\left\langle \left(\frac{\kappa\wedge dA-d\kappa\wedge A}{\kappa\wedge d\kappa}\right)^{2}\right\rangle.
\end{eqnarray}
Using Cartan's formula for the Lie derivative, $\mathcal{L}_{\xi}=\{\iota_{\xi},d\}$, and the fact that,
$$\iota_{\xi}A=\frac{A\wedge d\kappa}{\kappa\wedge d\kappa},$$
one finds,
\begin{eqnarray*}
(\mu,\mu)&=&\int_{X}\langle A\wedge dA\rangle-\int_{X}\frac{1}{\kappa\wedge d\kappa}\langle(\kappa\wedge dA)^{2}\rangle,\\
         &=&4\pi\overline{\operatorname{CS}}(A),
\end{eqnarray*}
as desired.  As in \cite[\S 3.3]{bw}, we observe that the path integral measure $\bar{\mathcal{D}}A$ should be identified with the corresponding symplectic measure $\operatorname{exp}(\Omega)$ in the path integral and we may write,
\begin{equation}\label{Anom1234}
\bar{Z}_{\mathbb{T}}(X,P,k)=\frac{e^{i k \operatorname{CS}_{X,P}(A_{P})}}{\operatorname{Vol}(\mathcal{G})}\left(\frac{-ik}{4\pi^2}\right)^{\Delta{\mathcal{G}/2}}\ddashint_{\bar{\mathcal{A}}_{P}}\exp\,\left[\Omega+\frac{i k}{4\pi}(\mu,\mu)\right].
\end{equation}
\section{Non-Abelian Localization for Abelian Chern-Simons Theory}\label{locabelchern}
Our starting point in this section is the main result of \S\ref{momsquared}, which expresses the abelian partition function in \emph{moment map squared} form,
\begin{equation}\label{Anom12345}
\bar{Z}_{\mathbb{T}}(X,P,k)=\frac{e^{i k \operatorname{CS}_{X,P}(A_{P})}}{\operatorname{Vol}(\mathcal{G})}\left(\frac{-ik}{4\pi^2}\right)^{\Delta{\mathcal{G}/2}}\ddashint_{\bar{\mathcal{A}}_{P}}\exp\,\left[\Omega+\frac{i k}{4\pi}(\mu,\mu)\right].
\end{equation}
Starting with this description of the partition function we follow the main arguments of \cite{bw}, which must be adapted slightly for an abelian structure group, to arrive at the final rigorous definition \ref{symrigdef}.
\subsection{A Two-Dimensional Description of the Abelian Partition Function}
Let,
\begin{eqnarray*}
f_{\mu}:=(\mu,\mu)&=&\int_{X}\langle A\wedge dA\rangle-\int_{X}\frac{1}{\kappa\wedge d\kappa}\langle(\kappa\wedge dA)^{2}\rangle,\\
         &=&4\pi\overline{\operatorname{CS}}(A).
\end{eqnarray*}
Observe that the critical points of $f_{\mu}$ satisfy the equation of motion,
\begin{equation}\label{eomshift}
dA-(\star_{H}dA)\wedge d\kappa-\kappa\wedge d\star_{H}dA=0.
\end{equation}
Recall the following,
\begin{define}
Let $(X,\phi,\xi,\kappa,\operatorname{g})$ be a contact metric three manifold and define the \emph{horizontal Hodge star} to be the operator,
\begin{equation*}
\star_{H}:\Omega^{q}(X,\frak{t})\rightarrow \Omega^{2-q}(H,\frak{t})\,\,q=0,1,2,
\end{equation*}
defined for $\beta\in\Omega^{q}(X,\frak{t})$ by,
\begin{equation}
\star_{H}\beta=\star(\kappa\wedge\beta),
\end{equation}
where $\star$ is the usual Hodge star operator on forms for the metric $\operatorname{g}=\kappa\otimes\kappa+d\kappa(\cdot,\phi\cdot)$ on $X$.
\end{define}
Since $f_{\mu}$ is invariant under the shift symmetry, it is clear that \eqref{eomshift} is also invariant under this symmetry. Thus, the critical points of $f_{\mu}$ can be classified as solutions of \eqref{eomshift} relative to any convenient gauge choice for the shift symmetry.  Observe that the quantity $\kappa\wedge dA$ transforms as,
$$\kappa\wedge dA\rightarrow \kappa\wedge dA+\sigma\kappa\wedge d\kappa,$$
under the shift symmetry for arbitrary $\sigma\in\Omega^{1}(X,\frak{t})$.  Thus, a valid gauge condition is given by setting,
\begin{equation}\label{shiftgauge1}
\star_{H} dA=0,
\end{equation}
since given arbitrary $A\in\Omega^{1}(X,\frak{t})$, we may set $\sigma:=\star_{H}dA$ uniquely so that $A-\sigma\kappa$ satisfies \eqref{shiftgauge1}.  In this gauge the solutions of the equation of motion \eqref{eomshift} are precisely the flat connections.
\begin{rem}
Note that the normalization $\bar{Z}_{\mathbb{T}}(X,P,k)$ in \eqref{Anom12345} needs to be revised slightly to take into account the fact that we have replaced the gauge group $\mathcal{G}$ with the group $\mathcal{H}:=\U(1)\ltimes \widetilde{\mathcal{G}}_{0}$ in our considerations.  In fact, we should formally replace $\mathcal{G}$ with $\mathcal{H}':=\U(1)\ltimes \widetilde{\mathcal{G}}$, where $\widetilde{\mathcal{G}}$ represents a central extension of the full gauge group by $\operatorname{U}(1)$.  We therefore wish to formally consider,
\begin{equation}\label{Anom123456}
\bar{Z}_{\mathbb{T}}'(X,P,k):=\frac{e^{i k \operatorname{CS}_{X,P}(A_{P})}}{\operatorname{Vol}(\mathcal{H}')}\left(\frac{-ik}{4\pi^2}\right)^{\Delta{\mathcal{H}'/2}}\ddashint_{\bar{\mathcal{A}}_{P}}\exp\,\left[\Omega+\frac{i k}{4\pi}(\mu,\mu)\right],
\end{equation}
where $\Delta_{\mathcal{H}'}=\operatorname{dim}\mathcal{H}'$.
As observed in \cite[Eq. 5.10]{bw} this results in a difference between $\bar{Z}_{\mathbb{T}}(X,P,k)$ and $\bar{Z}_{\mathbb{T}}'(X,P,k)$ by the finite multiplicative factor,
\begin{equation}\label{finfact}
\frac{\operatorname{Vol}(\mathcal{H}')}{i\operatorname{Vol}(\mathcal{G})}\left(\frac{-ik}{4\pi^{2}}\right)^{\frac{1}{2}(\Delta_{\mathcal{G}}-\Delta_{\mathcal{H}'})}=\operatorname{Vol}(U(1)^{2})\cdot\frac{4\pi^{2}}{k}.
\end{equation}
The technique of non-abelian localization applies more directly to $\bar{Z}_{\mathbb{T}}'(X,P,k)$ and this is the quantity that we will consider.  In the end we must multiply our results by the factor \eqref{finfact} to recover the Chern-Simons path integral $\bar{Z}_{\mathbb{T}}(X,P,k)$.
\end{rem}
Next, we focus on giving a ``two-dimensional'' description of the local symplectic geometry in $\bar{\mathcal{A}}_{P}$ around a critical point of the Chern-Simons action analogous to \cite[\S5.1]{bw}.  For this it will be convenient to choose a particular gauge for the shift symmetry corresponding to the gauge condition,
\begin{equation}\label{shiftgaugechoice}
\iota_{\xi}A=0.
\end{equation}
Clearly \eqref{shiftgaugechoice} defines a good gauge condition since if $\iota_{\xi}A=0$, then $\iota_{\xi}(A+\sigma\kappa)=\sigma\neq 0$, i.e. the condition \eqref{shiftgaugechoice} picks out a unique representative for each orbit of the shift symmetry.  Note that this gauge condition is defined on the tangent space $T_{A_{P}}\mathcal{A}_{P}$ about a flat connection $A_{P}$.  Of course, this is the natural gauge that we will implicitly work with throughout, and gives the identification,
\begin{equation}\label{natguage}
T_{A_{P}}\bar{\mathcal{A}}_{P}\simeq \Omega^{1}(H,\frak{t}),
\end{equation}
where $\Omega^{1}(H,\frak{t}):=\{A\in\Omega^{1}(X,\frak{t})\,\,|\,\,\iota_{\xi}A=0\}$.
Given that the Reeb vector field $\xi$ generates a locally free $\U(1)$ action, we naturally decompose the tangent space $T_{A_{P}}\bar{\mathcal{A}}_{P}\simeq \Omega^{1}(H,\frak{t})$ with respect to this action and write,
\begin{equation}\label{eigdecomp}
A=\sum_{l\in\Z}A_{l},
\end{equation}
where $A_{l}\in\Omega^{1}(H,\frak{t})$ are eigenmodes of the Lie derivative $\mathcal{L}_{\xi}$,
\begin{equation}\label{aeigdcom}
\mathcal{L}_{\xi}A_{l}=-2\pi i l A_{l}.
\end{equation}
We also decompose $Y\in\operatorname{Lie}(\mathcal{G})\simeq \Omega^{0}(X,\frak{t})$ with respect to the $\U(1)$ action,
\begin{equation}\label{yeigdecomp}
Y=\sum_{l\in\Z}Y_{l},
\end{equation}
where,
\begin{equation}\label{yeigdcom}
\mathcal{L}_{\xi}Y_{l}=-2\pi i l Y_{l}.
\end{equation}
Let,
\begin{equation}
\mathcal{L}:=X\times_{\U(1)}\C,
\end{equation}
denote the complex line V-bundle over $\Sigma$ associated to the standard representation on $\C$.  We view the eigenmodes $A_{l}\in\Omega^{1}(H,\frak{t})$ as naturally corresponding to elements of,
$$\Omega^{1}(\Sigma,\mathcal{L}^{l}\otimes\frak{t}):=\Gamma(\Sigma,T^{*}\Sigma\otimes\mathcal{L}^{l}\otimes\frak{t}),$$
and formally decompose the tangent space $T\bar{\mathcal{A}}_{P}$ at $A_{P}$ as,
\begin{equation}
T_{A_{P}}\bar{\mathcal{A}}_{P}=\bigoplus_{l\in\Z}\Omega^{1}(\Sigma,\mathcal{L}^{l}\otimes\frak{t}).
\end{equation}
Similarly, we decompose $\operatorname{Lie}(\mathcal{G})$,
\begin{equation}
\operatorname{Lie}(\mathcal{G})=\bigoplus_{l\in\Z}\Omega^{0}(\Sigma,\mathcal{L}^{l}\otimes\frak{t}).
\end{equation}
\subsection{Non-Abelian Localization Applied in Abelian Chern-Simons Theory}
As in \cite[\S 4.2]{bw}, a local symplectic neighborhood of $\bar{\mathcal{A}_{P}}$ near $\mathcal{M}_{P}$, say $N$, is an equivariant fibration,
\begin{equation}\label{locsympl}
F\rightarrow N\rightarrow \mathcal{M}_{P},
\end{equation}
where the fibre $F$ takes the form,
\begin{equation}\label{symbunf}
F=\mathcal{H}\times_{\mathcal{H}_{0}}(\frak{h}\ominus\frak{h}_{0}\ominus\mathcal{E}_{0}\oplus\mathcal{E}_{1}),
\end{equation}
and $\mathcal{H}_{0}$, $\mathcal{E}_{0}$, $\mathcal{E}_{1}$ remain to be identified.  Note that the symbol ``$\ominus$'' is to be interpreted in the sense of K-theory.  In \eqref{symbunf}, we have $\mathcal{H}=\U(1)\ltimes\tilde{\mathcal{G}}_{0}$ as before.  As in \cite[Eq. 5.27]{bw}, $\mathcal{H}_{0}$ is the subgroup of $\mathcal{H}$ that fixes $A_{P}$ and in general is of the form,
\begin{equation}\label{hnot}
\mathcal{H}_{0}=\operatorname{U}(1)\times\U(1)\times I,
\end{equation}
where $\operatorname{U}(1)$ arises as the central extension group for $\tilde{\mathcal{G}}_{0}$, $\U(1)$ arises as the group acting on $\mathcal{A}_{P}$ induced from the geometric action on $X$, and $I\simeq\mathbb{T}$ is the isotropy subgroup of $A_{P}$.  As in \cite[Eq. 5.29]{bw} we may identify $\mathcal{E}_{0}$ and $\mathcal{E}_{1}$ as,
\begin{equation}\label{idente0}
\mathcal{E}_{0}=\bigoplus_{\l\in\mathbb{N}}H^{0}_{\bar{\partial}}(\Sigma,(\mathcal{L}^{l}\oplus\mathcal{L}^{-l})\otimes\frak{t}),
\end{equation}
and,
\begin{equation}\label{idente1}
\mathcal{E}_{1}=\bigoplus_{\l\in\mathbb{N}}H^{1}_{\bar{\partial}}(\Sigma,(\mathcal{L}^{l}\oplus\mathcal{L}^{-l})\otimes\frak{t}).
\end{equation}
Now we express the partition function in \eqref{Anom12345},
\begin{equation}\label{ympartfun12}
\bar{Z}_{\mathbb{T}}(X,P,k)=\frac{e^{i k \operatorname{CS}_{X,P}(A_{P})}}{\operatorname{Vol}(\mathcal{G})}\left(\frac{-ik}{4\pi^2}\right)^{\Delta{\mathcal{G}/2}}\ddashint_{\bar{\mathcal{A}}_{P}}\exp\,\left[\Omega+\frac{i k}{4\pi}(\mu,\mu)\right],
\end{equation}
in a form that non-abelian localization can be more directly applied.  We write the following for $\bar{Z}_{\mathbb{T}}(X,P,k)$,
\begin{eqnarray}\label{ymeq123}
\frac{e^{i k \operatorname{CS}_{X,P}(A_{P})}}{\operatorname{Vol}(\mathcal{G})}\int_{\operatorname{Lie}\mathcal{G}}\left[\frac{dY}{2\pi}\right]\ddashint_{\mathcal{A}}\operatorname{exp}\left[\Omega+i\langle\mu,{Y}\rangle-\frac{4\pi i}{k}(Y,Y)\right],
\end{eqnarray}
where the equivalence of Eq.'s \eqref{ympartfun12} and \eqref{ymeq123} may be seen by doing the formal Gaussian integral over $\operatorname{Lie}\mathcal{G}$ in \eqref{ymeq123}.  We begin the localization computation by choosing $\theta$ as,
\begin{equation}
\theta := \tilde{J}df_\mu,
\end{equation}
where $f_{\mu}:=\frac{1}{2}(\mu,\mu)$.  We then write the integral in \eqref{ymeq123} over $N$ and this gives the following for $\bar{Z}_{\mathbb{T}}(X,P,k)$ in \eqref{locsympl} as,
\begin{equation}\label{ymeq1234}
\frac{4\pi^{2}}{k}\frac{\operatorname{Vol}(U(1)^{2})}{\operatorname{Vol}(\mathcal{H})}e^{i k \operatorname{CS}_{X,P}(A_{P})}\ddashint_{\frak{H}\times N}\left[\frac{dY}{2\pi}\right]\operatorname{exp}\left[\Omega+i\langle\mu,{Y}\rangle-\frac{4\pi i}{k}(Y,Y)+t\cdot D\theta\right],
\end{equation}
where $\frak{H}=\operatorname{Lie}\mathcal{H}$, $D:=d_{\operatorname{Lie}\mathcal{G}}$ denotes the formal equivariant derivative, $t\in\R$ and we include the normalization factor $\operatorname{Vol}(U(1)^{2})$ as in \eqref{finfact}.  The goal is now to reduce the integral over $\frak{H}\times N$ in \eqref{ymeq1234} to an integral over $\mathcal{M}_{P}$.  We start by observing that the fibre $F$ in \eqref{locsympl} may be modelled on the cotangent bundle $T^{*}\mathcal{H}$, so that $N$ equivariantly retracts onto a principal $\mathcal{H}$-bundle $P_{\mathcal{H}}$ over the moduli space $\mathcal{M}_{P}$.  Following the argument of \cite{w2} we observe that if $\mathcal{H}$ acts freely on $P_{\mathcal{H}}$ that the equivariant cohomology of the total space $P_{\mathcal{H}}$ may be identified with the ordinary cohomology of the quotient $P_{\mathcal{H}}/\mathcal{H}\simeq \mathcal{M}_{P}$.  This allows us to identify the equivariant forms $\Omega+i\langle\mu,{Y}\rangle$ and $(Y,Y)$ with the pullback of ordinary forms on $\mathcal{M}_{P}$.\\
\\
In Chern-Simons theory, $\mathcal{H}$ does not act freely on $N$, however.  We may still follow the same reasoning as in Yang-Mills theory \cite{w2} by taking into account that the subgroup $\mathcal{H}_{0}$ in \eqref{hnot} acts on $N$ with fixed points.  We then obtain an equivariant retraction $N_{0}$ of $N$ as a bundle with fibre $\mathcal{H}/\mathcal{H}_{0}$,
\begin{equation}\label{fundle}
\mathcal{H}/\mathcal{H}_{0}\rightarrow N_{0}\rightarrow \mathcal{M}_{P}.
\end{equation}
Following the same argument of \cite[Eq. 5.109]{bw}, we may identify the $\mathcal{H}$ equivariant forms,
$$\Omega+i\langle\mu,{Y}\rangle,\,\,\,(Y,Y),$$
on $N$ with the corresponding $\mathcal{H}_{0}$ equivariant forms,
$$\Omega+ia,\,\,\, n\Theta+pa,$$
on $\mathcal{M}_{P}$, respectively, via pullback by a map,
\begin{equation}\label{projeq}
\operatorname{pr}:N\rightarrow \mathcal{M}_{P}.
\end{equation}
Note that $p, a\in \frak{H}_{0}$, $\Theta\in H^{4}(\mathcal{M}_{P})$, $n=c_{1}(X)$ is the first Chern number of $X$ as a bundle over $\Sigma$, and we have abused notation by writing $\Omega$ to represent the corresponding forms on both $N$ and $\mathcal{M}_{P}$.  Note that for simplicity we assume that $\U(1)$ acts freely on $X$.  Let, $$\tilde{K}:=\frac{4\pi^{2}}{k}\frac{\operatorname{Vol}(U(1)^{2})}{\operatorname{Vol}(\mathcal{H})}e^{i k \operatorname{CS}_{X,P}(A_{P})}.$$
We may then write,
\begin{equation}\label{ymeq12345}
\bar{Z}_{\mathbb{T}}(X,P,k)=\tilde{K}\ddashint_{\frak{H}\times N}\left[\frac{dY}{2\pi}\right]\operatorname{exp}\left[\operatorname{pr}^{*}\Omega+ia\left(1-\frac{2\pi}{k}p\right)+\frac{2\pi i n}{k}\operatorname{pr}^{*}\Theta+t\cdot D\theta\right],
\end{equation}
exactly as in \cite[Eq. 5.111]{bw}.  Next, we decompose $\frak{H}$ as,
\begin{equation*}\label{frakH}
\frak{H}=\left(\frak{H}\ominus\frak{H}_{0}\right)\oplus\frak{H}_{0},
\end{equation*}
and integrate over the variables $a,p$ spanning $\frak{H}_{0}$ in \eqref{ymeq12345}.  Setting $\epsilon:=\frac{2\pi}{k}$, the integral over $a$ will produce a delta function that sets $p=\frac{1}{\epsilon}$, and we obtain the following for $\bar{Z}_{\mathbb{T}}(X,P,k)$,
\begin{equation}\label{ymeq123456}
\frac{\operatorname{Vol}(U(1)^{2})}{\operatorname{Vol}(\mathcal{H})}e^{i k \operatorname{CS}_{X,P}(A_{P})}\ddashint_{\left(\frak{H}\ominus\frak{H}_{0}\right)\times N}\left[\frac{dY}{2\pi}\right]\operatorname{exp}\left[\operatorname{pr}^{*}\Omega+i\epsilon n\operatorname{pr}^{*}\Theta+t\cdot D\theta|_{\left\{p=\frac{1}{\epsilon}\right\}}\right].
\end{equation}
Note that we will drop the term $\frac{1}{k}$ that occurs in the definition of $\tilde{K}$ above and implicitly redefine the partition function to take this into account.  As in \cite{w2}, the only term which is a not pull back from $\mathcal{M}_{P}$ is the localization term $t\cdot D\theta$, and we are left to perform the computation of the $t\cdot D\theta$ dependent part of \eqref{ymeq123456} over $F$.  In Yang-Mills theory \cite{w2} finds that the corresponding integral over $F=T^{*}\mathcal{H}$ produces a trivial factor of 1, and this is no longer the case in Chern-Simons theory as is observed in \cite{bw}.  The quantity of interest is then,
\begin{equation}\label{quantint}
I(\psi):=\frac{1}{\operatorname{Vol}\mathcal{H}}\ddashint_{\widetilde{F}}\left[\frac{dY}{2\pi}\right]\operatorname{exp}(t\cdot D\theta),
\end{equation}
where,
$$\widetilde{F}:=\left(\frak{H}\ominus\frak{H}_{0}\right)\times F,$$
and,
$$\psi=p+a\in\frak{H}_{0},$$
and we set $p=\frac{1}{\epsilon}$ at the end of the computation.  Following the exact same reasoning that leads to \cite[Eq. 5.117]{bw}, we may identify,
\begin{equation}\label{equivarcoh}
I(\psi)=\frac{1}{\operatorname{Vol}\mathcal{H}_{0}}\frac{e_{\mathcal{H}_{0}}(\mathcal{M}_{P},\mathcal{E}_{0})}{e_{\mathcal{H}_{0}}(\mathcal{M}_{P},\mathcal{E}_{1})},
\end{equation}
where $e_{\mathcal{H}_{0}}(\mathcal{M}_{P},\mathcal{E}_{0})$, $e_{\mathcal{H}_{0}}(\mathcal{M}_{P},\mathcal{E}_{1})$ are the $\mathcal{H}_{0}$-equivariant Euler classes of the bundles associated to $\mathcal{E}_{0}$, $\mathcal{E}_{1}$ as in \eqref{idente0}, \eqref{idente1} over $\mathcal{M}_{P}$.  This may be confirmed by direct computation exactly as in \cite[Appendix D]{bw} and we do not repeat this argument here.  Define,
\begin{equation}\label{epdef}
e(p):=\frac{e_{\mathcal{H}_{0}}(\mathcal{M}_{P},\mathcal{E}_{0})}{e_{\mathcal{H}_{0}}(\mathcal{M}_{P},\mathcal{E}_{1})}.
\end{equation}
Then our considerations so far yield,
\begin{equation}\label{absteq}
\bar{Z}_{\mathbb{T}}(X,P,k)=\frac{e^{i k \operatorname{CS}_{X,P}(A_{P})}}{\operatorname{Vol}I}\int_{\mathcal{M}_{P}}e(p)|_{\left\{p=\frac{1}{\epsilon}\right\}}\operatorname{exp}\left[\Omega+i\epsilon n\Theta\right].
\end{equation}
Note that in deriving \eqref{absteq} that the factor $\operatorname{Vol}(U(1)^{2})$ in \eqref{ymeq123456} cancels with a factor in $\operatorname{Vol}\mathcal{H}_{0}$ in $I(\psi)$.  Recall that $\Theta\in H^{4}(\mathcal{M}_{P})$ is
the cohomology class corresponding to the degree four element $(Y, Y)$
in the equivariant cohomology $H^4_{\mathcal{H}}(N)$.  We observe that in the abelian case $G=\mathbb{T}$ that $\Theta$ can also be described in terms of the universal bundle $\U$,
\begin{displaymath}
\xymatrix{\C \ar@{^{(}->}[r] & \U \ar[d]\\
                              & \operatorname{Jac}(\Sigma)\times\Sigma}.
\end{displaymath}
In other words,
\begin{equation*}
\Theta=-\frac{1}{2}c_{1}(\U)^{2}|_{\text{pt.}\in\Sigma},
\end{equation*}
where $\operatorname{Jac}(\Sigma)$ is the Jacobian of $\Sigma$.  In this case $\Theta=0$ since the universal bundle $\U$ for $\mathbb{T}$-bundles is the classical Poincar\'{e} line bundle, and the Poincar\'{e} line bundle is normalized to have degree $d=0$ when restricted to the Jacobian of $\Sigma$.  Thus, our computation produces the simple result,
\begin{equation}\label{absteq}
\bar{Z}_{\mathbb{T}}(X,P,k)=\frac{e^{i k \operatorname{CS}_{X,P}(A_{P})}}{\operatorname{Vol}I}\int_{\mathcal{M}_{P}}e(p)|_{\left\{p=\frac{1}{\epsilon}\right\}}\operatorname{exp}\left[\Omega\right].
\end{equation}
Our main goal now is to compute/define $e(p)|_{\left\{p=\frac{1}{\epsilon}\right\}}$.
\subsection{Computation of the Symplectic Abelian Partition Function}
We compute/define $e(p)|_{\left\{p=\frac{1}{\epsilon}\right\}}$ following the main arguments in \cite{bw} and adapt their technique to the case of an abelian structure group.  The main difference in the case of an abelian structure group shows up in the $k$-dependence of the partition function $\bar{Z}_{\mathbb{T}}(X,P,k)$.  Note that $\epsilon=\frac{2\pi}{k}$.  This difference is due to the fact that \cite{bw} works with irreducible flat connections and the corresponding zeroth cohomology spaces vanish, whereas an abelian structure group necessarily has non-vanishing zeroth cohomology.\\
\\
Recall,
\begin{equation}\label{epdef2}
e(p):=\frac{e_{\mathcal{H}_{0}}(\mathcal{M}_{P},\mathcal{E}_{0})}{e_{\mathcal{H}_{0}}(\mathcal{M}_{P},\mathcal{E}_{1})}.
\end{equation}
Let $\mathcal{E}_{0}^{[l]}$, $\mathcal{E}_{1}^{[l]}$ denote the natural eigenspaces in $\mathcal{E}_{0}$, $\mathcal{E}_{1}$ under the action of $\mathcal{H}_{0}$, so that,
\begin{equation}\label{ep1}
\mathcal{E}_{0}^{[l]}=H^{0}_{\bar{\partial}}(\Sigma,\mathcal{L}^{l}\otimes\frak{t}),
\end{equation}
and,
\begin{equation}\label{ep2}
\mathcal{E}_{1}^{[l]}=H^{1}_{\bar{\partial}}(\Sigma,\mathcal{L}^{l}\otimes\frak{t}).
\end{equation}
As in \cite[Eq. 5.126]{bw}, we may write,
\begin{equation}\label{ep123}
e(p)=\prod_{l\neq 0}\left[\frac{e_{\mathcal{H}_{0}}(\mathcal{M}_{P},\mathcal{E}_{0}^{[l]})}{e_{\mathcal{H}_{0}}(\mathcal{M}_{P},\mathcal{E}_{1}^{[l]})}\right]=\prod_{l\geq 1}\left[\frac{e_{\mathcal{H}_{0}}(\mathcal{M}_{P},\mathcal{E}_{0}^{[l]})\cdot e_{\mathcal{H}_{0}}(\mathcal{M}_{P},\mathcal{E}_{0}^{[-l]})}{e_{\mathcal{H}_{0}}(\mathcal{M}_{P},\mathcal{E}_{1}^{[l]})\cdot e_{\mathcal{H}_{0}}(\mathcal{M}_{P},\mathcal{E}_{1}^{[-l]})}\right],
\end{equation}
where $e_{\mathcal{H}_{0}}(\mathcal{M}_{P},\mathcal{E}_{0}^{[l]})$, $e_{\mathcal{H}_{0}}(\mathcal{M}_{P},\mathcal{E}_{1}^{[l]})$ denote the $\mathcal{H}_{0}$ equivariant Euler classes of the finite dimensional bundles determined by $\mathcal{E}_{0}^{[l]}$, $\mathcal{E}_{1}^{[l]}$ over $\mathcal{M}_{P}$.  \cite{bw} then finds a recursive relation between the equivariant Euler classes of $\mathcal{E}_{0}^{[l]}$, $\mathcal{E}_{0}^{[l]}$, $\mathcal{E}_{1}^{[-l]}$, $\mathcal{E}_{1}^{[-l]}$ by choosing a convenient holomorphic structure on $\mathcal{L}$.  Using this recursive relation \cite{bw} find,
\begin{equation}\label{epalfin}
e(p)=\prod_{l\neq 0}\frac{1}{e_{\mathcal{H}_{0}}(\mathcal{M}_{P},T\mathcal{M}_{P}^{l})}\left[e_{\mathcal{H}_{0}}(\mathcal{M}_{P},\mathcal{V}^{l}_{P})\right]^{nl},
\end{equation}
where $T\mathcal{M}_{P}^{l}$ denotes the $\mathcal{H}_{0}$ equivariant version of the tangent space of $\mathcal{M}_{P}$, and $\mathcal{V}^{l}_{P}$ denotes a bundle associated to the chosen holomorphic structure and is defined in \cite[pg. 103]{bw}.  The result in \eqref{epalfin} assumes that points in the moduli space correspond to \emph{irreducible} flat connections so that $\mathcal{E}_{0}^{[0]}$ vanishes.  In our case $\mathcal{E}_{0}^{[0]}=H^{0}_{\bar{\partial}}(\Sigma,\frak{t})$ does not vanish and we must revise accordingly.  Let $\mathcal{I}^{l}_{P}$ denote the $\mathcal{H}_{0}$ equivariant bundle over $\mathcal{M}_{P}$ associated to the bundle with fiber $\mathcal{E}_{0}^{[0]}$ such that $\mathcal{H}_{0}$ acts on the fiber with eigenvalue $-2\pi i l$.  Our revised version of \eqref{epalfin} is then,
\begin{equation}\label{epalfin2}
e(p)=\prod_{l\neq 0}\frac{e_{\mathcal{H}_{0}}(\mathcal{M}_{P},\mathcal{I}_{P}^{l})}{e_{\mathcal{H}_{0}}(\mathcal{M}_{P},T\mathcal{M}_{P}^{l})}\left[e_{\mathcal{H}_{0}}(\mathcal{M}_{P},\mathcal{V}^{l}_{P})\right]^{nl}.
\end{equation}
Following the the exact method leading to \cite[Eq. 5.144]{bw} we may factorize $e(p)$ as a product of the three terms,
\begin{equation}\label{epalfin23}
e(p)=\prod_{l\neq 0}\left[\prod_{j=1}^{\operatorname{dim}\mathbb{T}}(-ilp+\iota_{j})\right]\left[\prod_{j=1}^{\operatorname{dim}\mathcal{M}_{P}/2}\frac{1}{(-ilp+w_{j})}\right]\left[\prod_{j=1}^{\operatorname{dim}\mathbb{T}}(-ilp+\nu_{j})^{nl}\right], \end{equation}
where $\iota_{j}, w_{j}, \nu_{j}$ are the Chern roots of the bundles $\mathcal{I}_{P}$, $T\mathcal{M}_{P}$, $\mathcal{V}_{P}$, respectively.  Let,
\begin{eqnarray}
f_{\mathcal{I}}(z)&=&\prod_{l\neq 0}\prod_{j=1}^{\operatorname{dim}\mathbb{T}}(-il+z\iota_{j}),\\
f_{\mathcal{M}}(z)&=&\prod_{l\neq 0}\prod_{j=1}^{\operatorname{dim}\mathcal{M}_{P}/2}(-il+zw_{j})^{-1},\\
f_{\mathcal{V}}(z)&=&\prod_{l\neq 0}\prod_{j=1}^{\operatorname{dim}\mathbb{T}}(-il+z\nu_{j})^{nl},
\end{eqnarray}
where $z=1/p$ is a formal parameter.  Factoring out $p$ from each of the three terms in \eqref{epalfin23} and using the Riemann zeta function to define the infinite products,
\begin{eqnarray*}
\prod_{l\geq 1}p^{2\operatorname{dim}\mathbb{T}}=\operatorname{exp}(2\operatorname{dim}\mathbb{T}\cdot \operatorname{ln}p\cdot \zeta(0))=p^{-\operatorname{dim}\mathbb{T}},\\
\prod_{l\geq 1}p^{-\operatorname{dim}\mathcal{M}_{P}}=\operatorname{exp}(-\operatorname{dim}\mathcal{M}_{P}\cdot \operatorname{ln}p\cdot \zeta(0))=p^{\operatorname{dim}\mathcal{M}_{P}/2},
\end{eqnarray*}
we may write,
\begin{equation}\label{epalfin223}
e(p)=p^{\frac{1}{2}\left(\operatorname{dim}\mathcal{M}_{P}-2\operatorname{dim}\mathbb{T}\right)}\cdot f_{\mathcal{I}}(z)\cdot f_{\mathcal{M}}(z)\cdot f_{\mathcal{V}}(z).
\end{equation}
One may apply exactly the same reasoning that leads to \cite[5.167]{bw} using appropriate zeta and eta function regularizations to rigorously define the quantities $f_{\mathcal{I}}(z), f_{\mathcal{M}}(z), f_{\mathcal{V}}(z)$.  We will not repeat the full argument here and instead we point out the the main differences that arise for the case of an abelian group $\mathbb{T}$.  First, since $\mathcal{M}_{P}\simeq \mathbb{T}^{2g}$, we have,
\begin{eqnarray*}
c(\mathcal{M}_{P})&:=&c(T\mathcal{M}_{P}),\\
                          &=& \prod_{j=1}^{\operatorname{dim}\mathcal{M}_{P}/2}c(L_{j})=\prod_{j=1}^{\operatorname{dim}\mathcal{M}_{P}/2}(1+x_{j}),
\end{eqnarray*}
where $L_{j}=T\Sigma_{j}$, $x_{j}=c_{1}(L_{j})\in H^2(\Sigma_{j}, \Z)$, and $\Sigma_{j}\simeq(\operatorname{U}(1))^2$.  Since $\Sigma_{j}$ are Lie groups, the tangent bundles $T\Sigma_{j}$
are trivial and hence,
\begin{equation*}
x_{j}=c_{1}(T\Sigma_{j})=0.
\end{equation*}
Thus,
\begin{equation}
\widehat{A}(\mathcal{M}_{P})=\prod_{j=1}^{\operatorname{dim}\mathcal{M}_{P}/2}\frac{x_{j}/2}{\text{sinh}(x_{j}/2)}=1.
\end{equation}
Clearly, $c_{1}(T\mathcal{M}_{P})=0$ as well.  Using these observations and the method leading to \cite[5.167]{bw}, we have,
\begin{eqnarray}
f_{\mathcal{I}}(p)=\left(2\pi\right)^{\operatorname{dim}\mathbb{T}},\,\,\,
f_{\mathcal{M}}(p)=\left(\frac{1}{2\pi}\right)^{\operatorname{dim}\mathcal{M}_{P}/2},\,\,\,
f_{\mathcal{V}}(p)=\operatorname{exp}\left(\frac{-i\pi}{2}\eta_{0}\right),
\end{eqnarray}
where $\eta_{0}$ is the adiabatic eta invariant of $X$ and was first defined in \cite{nic}.  Note that $\eta_{0}$ is computed explicitly in \cite[Appendix C]{beas} and is given by,
\begin{equation*}
\eta_{0}=N\left(\frac{c_{1}(X)}{6}-2\sum_{j=1}^{M}s(\alpha_{j},\beta_{j})\right),
\end{equation*}
where $s(\alpha,\beta):=\frac{1}{4\alpha}\sum_{j=1}^{\alpha-1}\operatorname{cot}\left(\frac{\pi j}{\alpha}\right)\operatorname{cot}\left(\frac{\pi j\beta}{\alpha}\right)\in\Q$ is the classical Rademacher-Dedekind sum and $[g,n; (\alpha_{1},\beta_{1}),\ldots,(\alpha_{M},\beta_{M})]$ (for $\operatorname{gcd}(\alpha_{j},\beta_{j})=1$) are the Seifert invariants of $X$.  Overall, we have,
\begin{equation}\label{finep}
e(p)=\left(\frac{1}{2\pi}\right)^{\operatorname{dim}H^{1}(X,\frak{t})}\cdot k^{\frac{1}{2}\left(\operatorname{dim}H^{1}(X,\frak{t})-2\operatorname{dim}H^{0}(X,\frak{t})\right)}\cdot\operatorname{exp}\left(-\frac{i\pi}{2}\eta_{0}\right),
\end{equation}
where we have used the fact that $\operatorname{dim}H^{1}(X,\frak{t})=\operatorname{dim}\mathcal{M}_{P}$, $\operatorname{dim}H^{0}(X,\frak{t})=\operatorname{dim}\mathbb{T}$ and $p=\frac{1}{\epsilon}=\frac{k}{2\pi}$.\\
\\
Finally, as is noted in \cite[Pg. 89-92]{bw} the derivation of the partition function is done implicitly with respect to a choice of the so called \emph{Seifert framing} on $X$.  This choice of framing results in a difference of a factor of $e^{i\delta\Psi}$ in the partition function relative to the canonical framing, where for general gauge group, \cite[Eq. 5.101]{bw},
\begin{equation}\label{phaseshift}
e^{i\delta\Psi}=\operatorname{exp}\left(\frac{i\pi\Delta_{G}}{4}-\frac{i\pi\Delta_{G}\text{\v{c}}_{\frak{g}}}{12(k+\text{\v{c}}_{\frak{g}})}\theta_{0}+\frac{i\pi}{2}\eta_{0}\right).
\end{equation}
For the case of an abelian group $\mathbb{T}$, $\Delta_{G}=N$ and $\frac{i\pi\Delta_{G}\text{\v{c}}_{\frak{g}}}{12(k+\text{\v{c}}_{\frak{g}})}\theta_{0}=0$.  Consider the case that $X$ is given as the unit circle bundle defined by a smooth degree $d$ line bundle over $\mathbb{C}\mathbb{P}^{1}$.  In this case, one has \cite[Eq. 5.100]{bw},
\begin{equation}\label{etacase}
\eta_{0}=-\frac{d\Delta_{G}}{6}.
\end{equation}  
Plugging \eqref{etacase} into \eqref{phaseshift}, one then has,
\begin{equation}\label{phaseshift2}
e^{i\delta\Psi}=\operatorname{exp}\left(\frac{i\pi\Delta_{G}}{12}(3-d)\right).
\end{equation}
In general, if the framing of $X$ is twisted by $F\in\Z$ units, then $\operatorname{CS}(A^{\operatorname{g}})$ transforms by,
\begin{equation*}
\operatorname{CS}(A^{\operatorname{g}})\rightarrow \operatorname{CS}(A^{\operatorname{g}})+2\pi F.
\end{equation*}
The partition function $Z_{\mathbb{T}}(X,k)$ is then transformed by,
\begin{equation}\label{partform}
Z_{\mathbb{T}}(X,k)\rightarrow Z_{\mathbb{T}}(X,k)\cdot \operatorname{exp}\left(\frac{\pi i N F}{12}\right).
\end{equation}
Thus, \eqref{phaseshift2} shows that our computation has been done in a framing that is shifted by $F=3-d$ units away from the canonical framing. 
Plugging the result for $e(p)$ in \eqref{finep} into \eqref{absteq} we make the following,
\begin{define}\label{symrigdef}
Let $k\in \Z$, and let $X$ be a closed oriented three-manifold that admits a quasi-regular Sasakian structure $(\kappa,\Phi,\xi,\operatorname{g})$, with associated principal bundle structure,
\begin{displaymath}
\xymatrix{\xyC{2pc}\xyR{1pc}\U(1) \ar@{^{(}->}[r] & X \ar[d]\\
                              & \Sigma}.
\end{displaymath}
Define the \emph{symplectic abelian Chern-Simons partition function},
\begin{equation}
\bar{Z}_{\mathbb{T}}(X,k)=\sum_{[P]\in\operatorname{Tors}H^{2}(X,\Lambda)}\bar{Z}_{\mathbb{T}}(X,P,k),
\end{equation}
where,
\begin{equation}\label{prodef2}
\bar{Z}_{\mathbb{T}}(X,P,k)=k^{m_X}e^{i k \operatorname{CS}_{X,P}(A_{P})}e^{-\frac{i\pi}{2}\eta_{0}}\int_{\mathcal{M}_{P}}K_{X}\cdot\omega_{P},
\end{equation}
where, $m_X:=\frac{N}{2}(\operatorname{dim}H^{1}(X;\R)-2\operatorname{dim}H^{0}(X;\R))$, $K_{X}:=\frac{1}{|c_{1}(X)\cdot \prod_{i}\alpha_{i}|^{N/2}}$, $\omega_{P}:=\frac{\left(\sum_{j=1}^{g}d\theta_{j}\wedge d\bar{\theta}_{j}\right)^{gN}}{(gN)!(2\pi)^{2gN}}$.
\end{define}
\begin{rem}
Note that we assumed that $X$ was a principal $\U(1)$ bundle over a smooth surface $\Sigma$ in the above derivation of the partition function.  We accordingly revise the factor $1/\operatorname{Vol}(I)=1/|c_{1}(X)|$ in the partition function to $K_{X}$ in order to restore topological invariance in the case when orbifold points are present.
\end{rem}

\section*{Acknowledgments}
This work represents an extension of part of my Ph.D. thesis and accordingly there are many people whom I would like to thank here.  First and foremost, I would like to take this opportunity to thank my thesis advisor, Lisa Jeffrey, for her patience, wisdom, creativity and for sharing with me her encyclopedic knowledge of mathematics.  This work would not have been possible without her.  Among the many other people from whom I have benefited, I would particularly like to thank J\o rgen Andersen, Dror Bar-Natan, Chris Beasley, John Bland, Vincent Bouchard, Ben Burrington, Stanley Deser, Dan Freed, Benjamin Himpel, Roman Jackiw, Yael Karshon, Eckhard Meinrenken, Rapha\"{e}l Ponge, Fr\'{e}d\'{e}ric Rochon, Michel Rumin, Paul Selick, Nicolai Reshetikhin, Vladimir Turaev, Jonathan Weitsman, and Edward Witten.  This work was partially supported by a National Science and Engineering Research Council of Canada Graduate Scholarship Award, and the Danish National Research Foundation.



%

\end{document}